\documentclass{article}
\usepackage{spconf,amsmath,graphicx}

\usepackage{amssymb}
\usepackage{bm,amsmath,amssymb,amscd,graphicx,amsthm,booktabs,multirow}
\usepackage{algorithm}
\usepackage{algorithmic}
\usepackage{subfigure}
\usepackage{setspace}
\usepackage{microtype}
\usepackage[numbers,sort&compress]{natbib}
\usepackage{color}
\usepackage[bookmarks,colorlinks]{hyperref}

\usepackage{dsfont}
\usepackage{textcomp}
\usepackage{tabu}  

\usepackage{acronym}
\usepackage{textcomp}

\usepackage{endnotes}
\usepackage{comment}

\newcommand\edit[1]{{\color{black} #1}}

\acrodef{cs}[CS]{compressed sensing}	
\acrodef{ista}[ISTA]{iterative shrinkage thresholding algorithm}	
\acrodef{fista}[FISTA]{fast \ac{ista}}

\acrodef{dnn}[DNN]{deep neural network}
\acrodef{rnn}[RNN]{recurrent neural network}

\acrodef{lista}[LISTA]{learned \ac{ista}}	
\acrodef{alista}[ALISTA]{Analytic-LISTA}
\acrodef{adalista}[AdaLISTA]{Adaptive-LISTA}

\acrodef{amp}[AMP]{approximate message passing}

\acrodef{admm}[ADMM]{alternating direction method of multipliers}

\acrodef{snr}[SNR]{signal-to-noise ratio}

\acrodef{nmse}[NMSE]{normalized mean squared error}
\acrodef{2d}[2D]{two-dimensional}

\acrodef{doa}[DOA]{direction-of-arrival}
\acrodef{dft}[DFT]{discrete Fourier transform}

\usepackage{epstopdf}
\usepackage{xcolor}
\usepackage{mathtools}

\usepackage{mathrsfs}
\usepackage{pgfplots}
\usepackage{tikz}

\newtheorem{lemma}{Lemma}
\newtheorem{theorem}{Theorem}

\newtheorem{proposition}{Proposition}
\newtheorem{definition}{Definition}

\usetikzlibrary{arrows}

\allowdisplaybreaks[1]

\newcommand{\Supp}{{\textnormal{Supp}}}
\newcommand{\abs}[1]{\left\lvert#1\right\rvert}
\renewcommand{\paragraph}[1]{\noindent\textbf{#1}\quad}

\newcommand{\norm}[1]{\| #1 \|}

\title{Theoretical Linear Convergence of Deep Unfolding Network for Block-Sparse Signal Recovery}
\name{Rong Fu$^{\star}$ 
	\qquad Vincent Monardo$^{\dagger}$
	\qquad Tianyao Huang$^{\star}$ 
	\qquad Yimin Liu$^{\star}$
	\thanks{
		Part of this paper was presented in part at the IEEE International Conference on Acoustics, Speech and Signal Processing (ICASSP), Toronto, Ontario, Canada, June 2021 \cite{2021Deep}.
	}
}

\address{$^{\star}$ Tsinghua University, Beijing, China \\
	$^{\dagger}$ Carnegie Mellon University, Pittsburgh, PA, USA
}

\begin{document}
	%
\maketitle
\begin{abstract}

In this paper, we consider the recovery of the high-dimensional block-sparse signal from a compressed set of measurements, where the non-zero coefficients of the recovered signal occur in a small number of blocks.
Adopting the idea of deep unfolding, we explore the block-sparse structure and put forward a block-sparse reconstruction network named Ada-BlockLISTA, which performs gradient descent on every single block followed by a block-wise shrinkage.
Furthermore, we prove the linear convergence rate of our proposed network, which also theoretically guarantees exact recovery for a potentially higher sparsity level based on   underlyingblock structure.
Numerical results indicate that Ada-BlockLISTA yields better signal recovery performance compared with existing algorithms, which ignore the additional block structure in the signal model.
\end{abstract}
	\begin{keywords}
		Compressed sensing, block-sparse, deep unfolding, learned ISTA, Adaptive-LISTA, harmonic retrieval problem, linear convergence
	\end{keywords}
	\section{Introduction}
	\label{sec:intro}

	As it is crucial to minimize sample size required in estimation problems, compressive sensing (CS) plays an indispensable role in the field of signal processing \cite{4313110,review,Eldar2012}. 
	
	There are many existing problems which can be formulated in terms of \ac{cs}, including sparse channel estimation \cite{5621984}, beam pattern synthesis \cite{7952786}, \ac{doa} estimation \cite{Balakrishnan2004A,Compbeamf} and range-Doppler estimation \cite{HuangCSFAR}.
	The recovery is completed by finding a sparse representation of the time-domain signal over a dictionary matrix $\bm \Phi \in \mathbb{C}^{N \times M}$, which is always a row-sampled \ac{dft} matrix.

	Since the ill-posedness of this kind of problem, a myriad of methods a myriad of 
	One of the well-known $\ell_1$-norm regularization techniques is \ac{ista} \cite{Beck2009A}, which is often-used 
	in a wide range of applications \cite{Combettes2006Signal, antonello2018proximal}.  
	Since ISTA is solving a convex optimization problem, it is guaranteed to converge to a solution under the correct circumstances. 
	As an iterative solver, it often costs too much time for \ac{ista} to converge, thus difficult to apply in various real-time applications. Some pre-defined optimization parameters such as the step size and regularization parameter are also required to set carefully based on prior knowledge, which may become quite a challenge in some cases.\cite{ZhangG17b}.  

	Recently, there has been an outburst of studies solving this problem via \acp{dnn}. \acp{dnn} are able to approximate nonlinear functions between inputs and outputs, which motivates researchers to consider the possibilities of finding the best possible solution within a limited calculation time.
	Since the remarkable success of \acp{dnn} in a variety of applications, it motivates us to use \acp{dnn} in sparse linear inverse problems, which significantly improve in both accuracy and complexity over traditional algorithms.
	From this perspective, Gregor and LeCun \cite{Gregor2010Learning} have proposed a \ac{rnn} to solve sparse coding problems, named \ac{lista}.
	Based on the iterative structure of \ac{ista}, the authors free the traditional parameters in \ac{ista} to data-driven variables 
	and unfold \ac{ista} algorithms into a \ac{rnn} structure. 
	\ac{lista} networks show improved performance in terms of convergence speed in both theoretical analysis and empirical results \cite{AMP-Inspired,OnsagerLAMP,Fu2019}.

	One drawback of \ac{lista} is that it is a hard-coded network which can only be used for a fixed dictionary.
	Many works have made some modifications in \ac{lista} to increase its adaptability by embedding the dictionary in the network architecture.
	\cite{NEURIPS2018_cf8c9be2} couples the learned matrices in \ac{lista} by embedding the dictionary matrix into the network structure.
	Base on the work of \cite{NEURIPS2018_cf8c9be2}, robust-ALISTA is proposed in \cite{liu2018alista} which explicitly calculates the learned matrices by solving a coherence minimization problem.
	As robust-ALISTA only needs to learn a few network parameters, like step size and threshold, it gains some robustness against stochastic model permutations when using an end-to-end robust training strategy. 
	Furthermore, there is another adaptive variation of \ac{lista}, named \ac{adalista} \cite{2020Ada}, which generalizes the application of robust-ALISTA to various model scenarios.
	Beyond enjoying \edit{accelerated} convergence speed, \ac{adalista} 
	is able to serve different dictionaries using the same weight matrix without retraining the whole network.
	
	However, all these LISTA-type networks are sensitive to the signal sparsity level.
	While increasing the number of non-zero elements in the sparse signal of interest, LISTA and its variants as well as other sparsity-exploiting methods will suffer decreased signal recovery accuracy.
	In practical applications, the recovered signal may not be sufficiently sparse, such as block-sparse problems which recover an unknown, block-sparse signal, where the nonzero elements are distributed in blocks.
	Block-sparsity naturally arises in many applications such as range-Doppler reconstruction of extended targets \cite{ExtendedTarget}, estimation and equalization of sparse communication channels \cite{990897}, sensing resources in compressed DNA micro-arrays \cite{4550564}, color imaging \cite{Majumdar2010Compressed}, multiple measurement vectors, etc.
	
		Based on \ac{adalista}, a block-sparse reconstruction network, named Ada-BlockLISTA, has been proposed in \cite{2021Deep}.
		By performing gradient descent and soft threshold individually for each separate block, Ada-BlockLISTA makes good use of the natural structure of a block-sparse signal and is able to recovery its non-zeros blocks correctly, while original \ac{adalista} fails especially with a large number of non-zeros blocks. 
	In this paper, we further derive rigorous theoretical analysis for this structured network. 
	extend
	In this paper, we aim to extend the applicability of LISTA-type networks to block-sparse recovery problems. 
	We first formulate the signal model with block structure, and introduce relevant iterative algorithms solving the block-sparse recovery problem in Section \ref{sec:basicsolution}. 
	By exploring the block structure of the signal model, we review prior work of unfolded networks and propose our Ada-BlockLISTA network in Section \ref{sec:Proposed}, which is an extension of \ac{adalista}.
	%
	In Section \ref{sec:Convergence}, we further analyze the convergency of the proposed network and also demonstrate its convergence acceleration over traditional iterative methods by extensive simulations.
    In numerical simulations, we apply our Ada-BlockLISTA to \ac{2d} harmonic retrieval in Section \ref{sec:Numerical}, which shows great advantage of recovery performance versus its non-learned counterparts.
	
	
	
	\textit{Notation}: 
	The symbol $\mathbb{C}$ denotes the set of complex numbers. Correspondingly, $\mathbb{C} ^{M}$ and $\mathbb{C} ^{M  \times N}$ are the sets of the $M$-dimensional ($M$-D) vectors and $M \times N$ matrices of complex numbers, respectively.
	The subscripts $[\cdot]_{i}$ and $[\cdot]_{i,k}$ denote the $i$-th entry of a vector and the $i$-th row, $k$-th column entry of a matrix. 
	We let $[\cdot]$ and $\{\cdot\}$ denote a vector/matrix and a set, respectively. 
	We use a set in subscript to construct a vector/matrix or set, e.g., for a set $\mathcal{N}:=\{1,2,\dots,N-1\}$ and vectors $\bm x_n \in \mathbb{C}^{M}$, $n \in \mathcal{N}$, $[\bm x_n]_{n \in \mathcal{N}}$ and $\{\bm x_n\}_{n \in \mathcal{N}}$  representing the matrix $[\bm x_1, \bm x_2,\dots,\bm x_{N}] \in \mathbb{C}^{M \times N}$ and the set $\{\bm x_1, \bm x_2,\dots,\bm x_{N}\}$, respectively.
	The transpose and Hermitian transpose are denoted by the superscripts $(\cdot)^T$ and $(\cdot)^H$, respectively.
	For a vector, $\| \cdot \|_0$  and $\| \cdot \|_q$ denote the $\ell_0$ ``norm'' and $\ell_q$ norm, $q \ge 1$, respectively. We define $\Supp(\cdot)$ as the support of a vector.
	
	$
	$

	\section{System Model and Traditional Iterative Algorithms}
	\label{sec:basicsolution}
	
	In this section, we introduce the signal model with block-sparsity.
	Recall that traditional iterative algorithms such as \ac{ista} leverage the assumption of sparsity in the ground truth signal.
	We also extend \ac{ista} to Block-ISTA, which recovers the block-sparse signals through $\ell_{2,1}$ norm minimization.
	
	\subsection{System Model}
	\label{ssec:blocksystem}
	
	Our goal is to recover an unknown, sparse signal from noisy observations taken from a known, under-determined dictionary $\bm \Phi \in \mathbb{C}^{N \times M}$, with $ N < M$. 
	In this ill-conditioned problem, noisy observation $ \bm y \in {\mathbb{C}^N}$ can be formulated as
	
	\begin{equation}  \label{eq:system}
	\bm y =\bm \Phi \bm x^* + \bm \varepsilon,
	\end{equation}
	where $ \bm x^* \in {\mathbb{C}^M}$ is the ground truth, and $ \bm \varepsilon \in {\mathbb{C}^N}$ is additive random noise present in the system. 
	Next, we define block-sparsity according to \cite{Block-Yonina}.
	
	Suppose that a block-sparse signal $\bm x$ is divided into $Q$ blocks of length $P$, as shown in \eqref{eq:blockx}. 
	Inside each block, there are are $P$ elements, denoted $ x_{p,q} \in \mathbb{C}$, to construct sub-vector ${\bm x}_q \in \mathbb{C}^{P}$, i.e., ${\bm x}_q :=\left[ x_{p,q} \right]_{ 1 \le p \le P}^T$.
	\begin{eqnarray}\label{eq:blockx}
	\bm x = {[
		\underbrace { x_{1,1} \cdots x_{P,1} }_{{\bm x}_1^T}\;
		\underbrace { x_{1,2} \cdots x_{P,2} }_{{\bm x}_2^T}\;\cdots\;
		\underbrace { x_{1,Q} \cdots x_{P,Q} }_{{\bm x}_Q^T}]^T}.
	\end{eqnarray}
	If there are at most $K$ nonzero blocks ($K \ll Q$) in $\bm x$, it is called $K$-block-sparse.
	Sharing the same nested structure with $\bm x$, the dictionary matrix $\bm \Phi$ is also divided into $Q$ blocks, i.e.,
	\begin{eqnarray}\label{eq:blockdict}
	\bm \Phi = [
	\underbrace {\bm \phi _{1,1} \cdots \bm \phi _{P,1}}_{{\bf{\Phi }}_1}\;
	\underbrace {\bm \phi _{1,2} \cdots \bm \phi _{P,2}}_{{\bf{\Phi }}_2}\;\cdots\;
	\underbrace {\bm \phi _{1,Q} \cdots \bm \phi _{P,Q}}_{{\bf{\Phi }}_Q}],
	\end{eqnarray}
	where each sub-matrix ${{\bf{\Phi }}_q} \in {\mathbb{C}^{N \times P}}$ has $P$ columns and $\bm \phi _{p,q}$ denotes the $p$-th column of ${{\bf{\Phi }}_q}$.
	A visualization of the block-sparse model is shown in Fig.~\ref{fig:model}.
	\begin{figure} [tb]
		\centering
		\subfigure[]{ 
			\label{fig:model}
			\includegraphics[width=0.35\columnwidth]{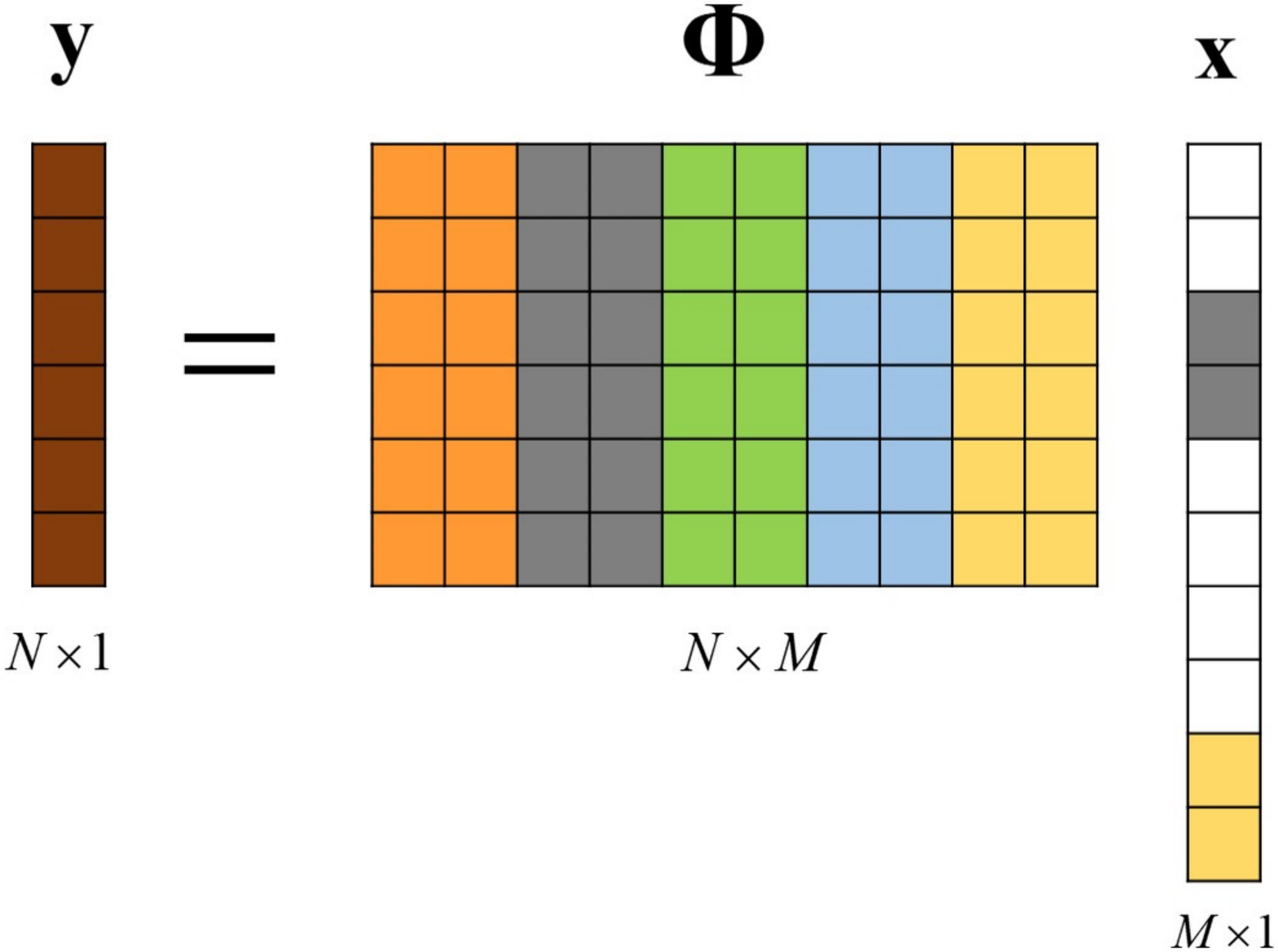}
		}
		\hspace{.9in}
		\subfigure[]{ 
			\label{fig:gram}
			\includegraphics[width=0.2\columnwidth]{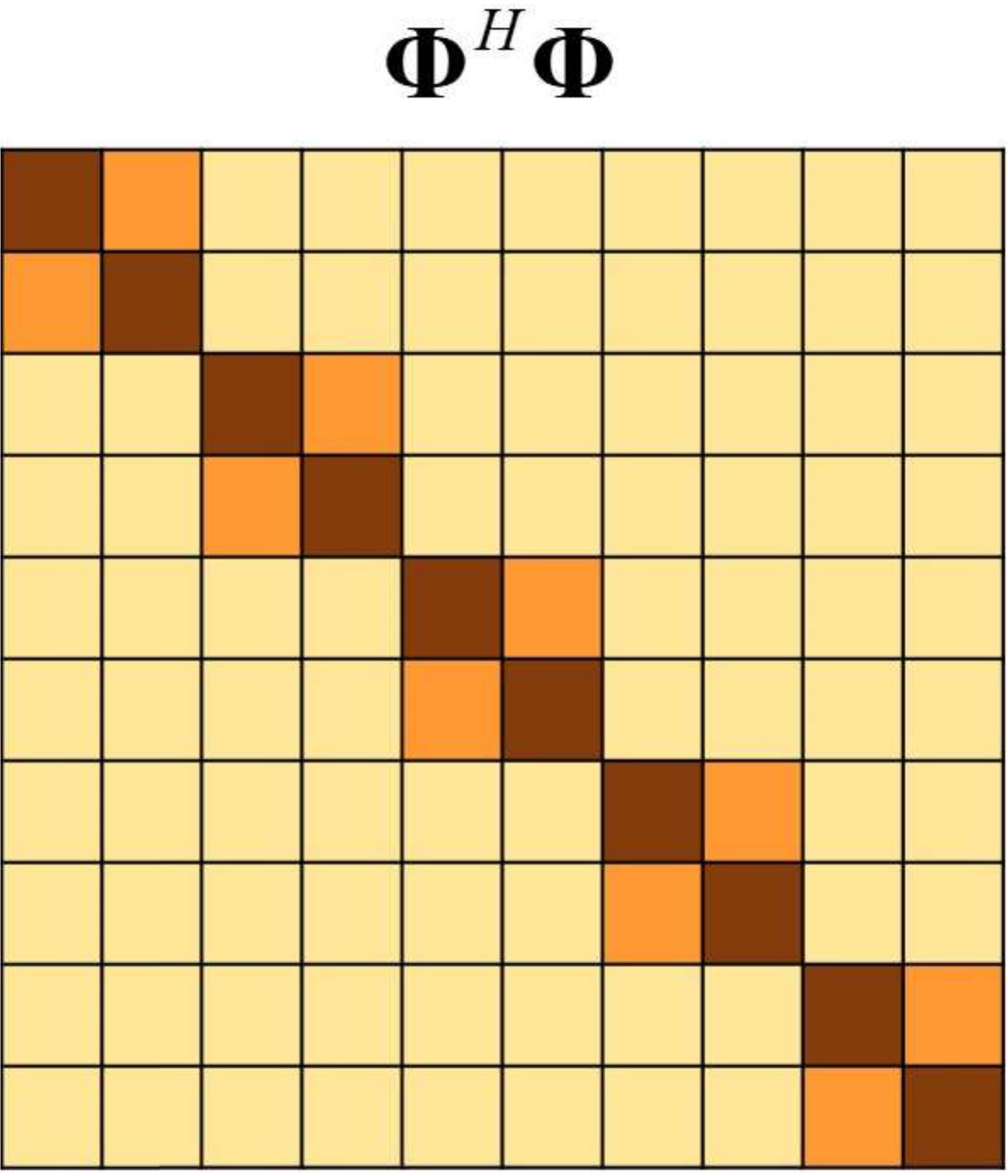}
		}
		\caption{(a) A visualization of the block-sparse signal model where each color in $\bf{\Phi}$ and $\bm x$ corresponds to a different block. White corresponds to zero-entries in $\bm x$. (b) Visualization of the Gram matrix ${\bm \Phi}^H{\bm \Phi}$ with block-structure.
		}
		
	\end{figure}

	\subsection{Traditional Iterative Algorithms}
	\label{sec:Traditional}

	In this subsection, we briefly review \ac{ista} and Block-ISTA, which are designed to solve sparse and block-sparse problems, respectively.

	In the framework of \ac{cs}, the following $\ell_1$-regularized regression is a well-known technique to harness prior knowledge of sparsity, given by
	\begin{equation} \label{eq:sparseregression}
	\mathop { \min }\limits_{\bm x} \frac{1}{2}\left\| {{\bm y - \bm \Phi \bm x}} \right\|_2^2 + \lambda {\left\| \bm x \right\|_1},
	\end{equation}
	where $\lambda$ is a regularization parameter controlling the sparsity penalty characterized by minimizing the $\ell_1$ norm.

	Standard \ac{ista} solves the above optimization problem
	by iteratively performing Lipschitz gradient descent with respect to the measurement error \cite{Beck2009A,Wu2020Sparse}. 
	Specifically, the sparse solution in the $(t+1)$-th iteration, denoted by $\bm x^{(t+1)}$, is pursued by the following recursion:

	\begin{equation}\label{eq:ISTA}
	\begin{aligned}
	\bm x^{(t+1)}  = {\eta_{\frac{\lambda }{L}}}\left( {\frac{1}{L}{\bm \Phi ^H}\bm y + \left( {\bm I - \frac{1}{L}{\bm \Phi ^H}\bm \Phi } \right){\bm x^{(t)}}} \right),\\
	\end{aligned}
	\end{equation}
	where 
	$L$ is the Lipschitz constant, given by $L = {\lambda _{\max }}( {{\bm \Phi ^H}\bm \Phi } )$, and $\lambda _{\max }(\cdot)$ represents the maximum eigenvalue of a Hermitian matrix. The element-wise soft-threshold operator $\eta$ is defined as

	\begin{equation}\label{eq:Soperator}
	\left[ {\eta_{\theta} }\left( \bm{u} \right) \right]_i= \text{sign} \left( {{[\bm{u}]_i}} \right)\left( {\left| {{[\bm{u}]_i}} \right| - \theta} \right)_ + ,
	\end{equation}
	where sign$ (\cdot)$ returns the sign of a scalar, $ (\cdot)_+ $ 
	gets the positive part of a scalar, 
	and $ \theta > 0 $ is the threshold.
	
	While the signal of interest $\bm x$ possesses block-sparse structure, 
	many researchers have designed specialized algorithms in principle to utilize the block-sparse structure, such as block-OMP \cite{bomp}, block-based CoSaMP \cite{2010BCoSaMP} and block-sparse Bayesian learning \cite{BlockBayesian}. 
	To compute an estimate which maintains a block-sparse structure, \eqref{eq:sparseregression} turns into the following mixed-norm optimization problem based on the partition of blocks,
	
	\begin{equation} \label{eq:blocksparseregression}
	\mathop { \min }\limits_{\bm x} \frac{1}{2}\left\| \bm y - \bm \Phi \bm x \right\|_2^2 + \lambda \left\|\bm x\right\|_{2,1},
	\end{equation}
	where the $\ell_{2,1}$ norm is defined as 
	$\left\|\bm x\right\|_{2,1} = \sum\limits_{q = 1}^Q { \left\|\bm x_q\right\|_{2} } = \sum\limits_{q = 1}^Q {\sqrt {\sum\limits_{p = 1}^P {{x_{p,q}^2}} } }.$

In principle, \ac{ista} can be extended in the block-sparse setup to derive an algorithm named Block-ISTA, which performs two steps (gradient descent and soft-thresholding) for every block $q \in [1,Q]$ individually as follows:

\begin{equation}\label{eq:Block-ISTA}
\begin{aligned}
\bm z_q^{(t+1)}  
& = {\bm x}_q^{(t)} + \frac{1}{L}{\bm \Phi_q ^H} \left(\bm y  - \bm \Phi {\bm x^{(t)}} \right),\\
{\bm x}_q^{(t+ 1)} 
& = {\bm z}_q^{(t+ 1)}{\left( 1 - {\theta} / \left\| {\bm z}_q^{(t+ 1)} \right\|_2 \right)_{ + }},
\end{aligned}
\end{equation}
where the threshold $ \theta > 0 $ is compared with the $\ell_{2}$ norm of each block, rather than the absolute value of each element in $\bm x$. 
Comparing \eqref{eq:ISTA} with \eqref{eq:Block-ISTA}, one can see that the same two-step process is being taken, where the second step in \eqref{eq:Block-ISTA} is a block-wise soft-thresholding operation derived from the proximal operator of the $\ell_{2,1}$ norm.
This thresholding operation forces the updated signal ${\bm z}^{(t+ 1)}$ in the previous step to be block-sparse: we use block-wise soft thresholding and set blocks in ${\bm z}^{(t+ 1)}$ to $\bm 0$ if its $\ell_{2}$ norm of the block ${\bm z}_q^{(t+ 1)}$ is less than $ {\theta} $.

\subsection{Brief Review of Learned \ac{lista} Networks}
\label{ssec:Ada-LISTA}

\ac{ista} demonstrates considerable accuracy on recovering sparse signals, but takes a lot of time to converge \cite{Draganic2017On}.
As an unfolded version of \ac{ista} iterations, \ac{lista} is a \ac{rnn} containing only $T$ layers, where each layer is corresponding to an \ac{ista} iteration which computes an estimate as follows \cite{Gregor2010Learning}.
\begin{equation}\label{eq:LISTA}
\bm{x}^{(t + 1)} = \eta_{{\theta}^{(t)}}\left( {\bm W_e}\bm y + {\bm W_g}{ \bm{x}^{(t)} }  \right).
\end{equation}
Here, the terms ${\lambda /L}$, ${\frac{1}{L}{\bm \Phi ^H}\bm y}$ and ${\left( {\bm I - \frac{1}{L}{\bm \Phi ^H}\bm \Phi } \right)}$ in \eqref{eq:ISTA} are replaced by learned variables ${\theta}^{(t)}$, ${\bm W_e}\in {\mathbb{C}^{M \times N}}$ and ${\bm W_g}\in {\mathbb{C}^{M \times M}}$, respectively. 
The matrices ${\bm W_e} $ and ${\bm W_g} $ are named the \emph{filter matrix} and the \emph{mutual inhibition matrix}, respectively, which are learned from training data.

However, once the learning process is completed, the learned network variables ${\bm W_g}$ and ${\bm W_e}$ are tailored for the specific dictionary matrix $ \bm \Phi$. 
In general, \ac{lista} requires that during the inference stage, the test signals be drawn from the same distribution as the training samples. 
That is to say, a trained \ac{lista} fails to recover signals which are not generated from the underlying dictionary in the training data. 

To build a more generalized network for various dictionary choices,
\ac{adalista} has been proposed to increase the adaptability of this deep unfolding technique \cite{2020Ada}. 
One iteration of \ac{adalista} is defined as,
\begin{eqnarray}\label{eq:adalista}
\bm{x}^{(t + 1)} = \eta_{{\theta}^{(t)}}\left( {\gamma}^{(t)}{\bm \Phi}^H\!\!{\bm W_2}\!\!^H\!\! \bm y + \left( {\bm I} - {\gamma}^{(t)}{\bm \Phi}^H\!\!{\bm W_1}\!\!^H\!\!{\bm W_1}\!\!{\bm \Phi} \right) { \bm{x}^{(t)} }  \right),
\end{eqnarray}
where ${\bm W_1}, {\bm W_2}\in {\mathbb{C}^{N \times N}}$ are shared across different layers, while ${\theta}^{(t)}$ and ${\gamma}^{(t)}$ are the soft threshold and the step size of the $t$-th layer. 

Comparing \eqref{eq:LISTA} and \eqref{eq:adalista}, we can say that ${\bm W_e}$ in \ac{lista} corresponds to ${\gamma}^{(t)}{\bm \Phi}^H{\bm W_2}^H$ in \ac{adalista}, while ${\bm W_g}$ in \ac{lista} corresponds to $\left( {\bm I} - {\gamma}^{(t)}{\bm \Phi}^H{\bm W_1}^H{\bm W_1}{\bm \Phi} \right)$ in \ac{adalista}. Instead of using the fixed network variables ${\bm W_g}$ and ${\bm W_e}$, the learned \emph{filter matrix} and \emph{mutual inhibition matrix} in \ac{adalista} make use of information in dictionary matrix ${\bm \Phi}$, so that it can adapt to various models.
Furthermore, \ac{adalista} with a single weight matrix is defined by
\begin{equation}
\label{eq:adalista_singleweight}
\bm{x}^{(t + 1)} = \eta_{{\theta}^{(t)}}\left( \bm{x}^{(t)} + {\gamma}^{(t)}{\bm \Phi}^H {\bm W_2}^H \left( \bm y - {\bm \Phi} \bm{x}^{(t)} \right) 
\right),
\end{equation}
which can be viewed as a reduced version of \eqref{eq:adalista} if we have $ {\bm W_2} = {\bm W_1}^H{\bm W_1}$.
The single weight matrix instantiation of AdaLISTA is convenient for proving its theoretical performance guarantees.

In the next section, we further utilize specific knowledge of block structure in \ac{adalista} to construct a block-sparse reconstruction network.

\section{Proposed Block-Sparse Recovery Network}
\label{sec:Proposed}

In this section, we outline our proposed Ada-BlockLISTA network. Additionally, we outline the training and testing details in order to utilize Ada-BlockLISTA in experimental applications.

\subsection{Proposed Network Structure}
\label{ssec:Ada-BlockLISTA}

Motivated by $\ell_{2,1}$ minimization method such as Block-ISTA, we explore the block structure of the signal model, and propose our Ada-BlockLISTA network, which is derived from \ac{adalista}.
Fig.~\ref{fig:frameLISTA} illustrates the network structure of both \ac{adalista} and Ada-BlockLISTA.

\begin{figure}[t]
	\centering 
	\subfigure[]{ 
		\includegraphics[width=0.4\columnwidth]{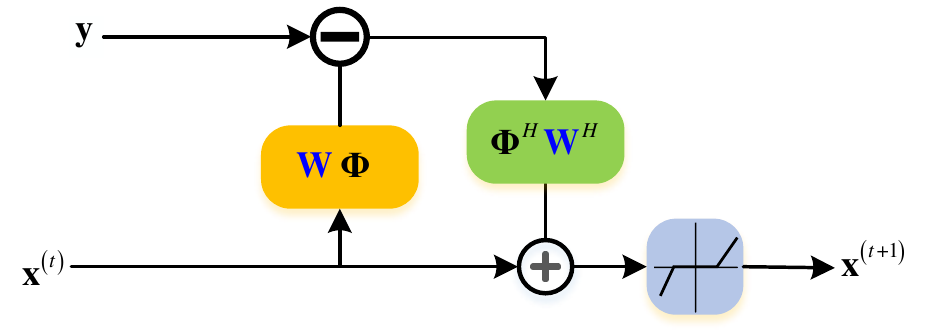}
	} 
	\subfigure[]{ 
		\includegraphics[width=0.4\columnwidth]{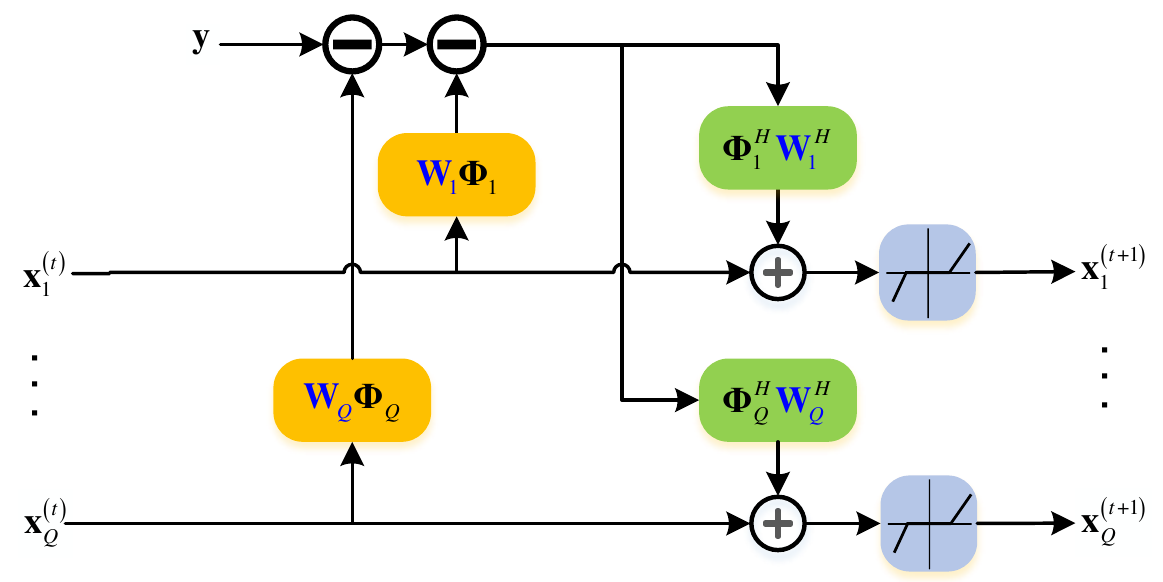}
	} 
	\caption{The Block diagram of (a) \ac{adalista} and (b) Ada-BlockLISTA architecture in one layer.}
	\label{fig:frameLISTA}
\end{figure}

As ${\bm \Phi}$ and $\bm{x}^*$ share the same block structure, at any iteration $t$, we can rewrite $ {\bm \Phi} \bm{x}^{(t)},$ as  $\sum\limits_{i = 1}^Q {{{\bm \Phi}_i}{\bm x}_i^{(t)}} $ with respect to the estimate at the $t$-th iteration, $\bm{x}^{(t)}$.
In Block-ISTA, each individual block in $\bm x$ is updated via a gradient projection and soft threshold step, independently. 
Thus, we learn an individual weight matrix ${\bm W}_q$ for each $q$-th block.
For the nonlinear shrinkage function, we also encourage block sparsity by applying soft thresholding to each block individually. 
At the $t$-th layer of Ada-BlockLISTA, the update rule of the $q$-th block in the block-sparse signal is formulated as,
\begin{subequations}\label{eq:AdaBlockLISTA}
	\begin{align}
	{\bm z}_q^{(t+ 1)} 
	& = {\bm x}_q^{(t)} + {\gamma}^{(t)}{\bf{\Phi }}_q^H{\bm W}_q^H {\bm r^{(t)}}, \label{eq:AdaBlockLISTA-sub1}\\
	{\bm x}_q^{(t+ 1)} 
	& = {\bm z}_q^{(t+ 1)}{\left( 1 - {\theta ^{(t)}} / \left\| {\bm z}_q^{(t+ 1)} \right\|_2 \right)_{ + }},
	\label{eq:AdaBlockLISTA-sub2}
	\end{align}
\end{subequations}
where ${\bm r^{(t)}} = \bm y - \sum\limits_{i = 1}^Q { {{\bm \Phi}_i}{\bm x}_i^{(t)}} $ is the residual at the $t$-th layer.
Consistent with previous algorithms discussed, Ada-BlockLISTA consists of two main steps.
First, we perform block-wise gradient descend and update the value of each block according to the current residual value ${\bm r^{(t)}}$. 
Then, to force the updated signal ${\bm z}_q^{(t+ 1)}$ in the previous step to be block-sparse, we use block-wise soft thresholding and set blocks in ${\bm z}_q^{(t+ 1)}$ to $\bm 0$ if its $\ell_{2}$ norm is less than $ {\theta}^{(t)} $.

If we force the matrices $\{{\bm W}_q\}_{q=1}^{Q}$ to be the same across different blocks, the term ${\bf{\Phi }}_q^H{\bm W}_q^H$ in \eqref{eq:AdaBlockLISTA} corresponds to the $q$-th block of ${\bm \Phi}^H {\bm W_2}^H$
in \eqref{eq:adalista_singleweight}.
Moreover, comparing \eqref{eq:adalista_singleweight} and \eqref{eq:AdaBlockLISTA}, \ac{adalista} and Ada-BlockLISTA use the soft threshold parameters $\{ {\theta}^{(t)} \}_{t=0}^{T}$ in different ways.
In \ac{adalista} with a single weight, we perform element-wise soft thresholding and set each element in ${\bm x}^{(t)}$ with an absolute value less than $ {\theta}^{(t)} $ to $0$.
On the other hand, we use $ {\theta}^{(t)} $ as a block-wise threshold in Ada-BlockLISTA, which is compared with the $\ell_{2}$ norm of every single block ${\bm z}_q^{(t)}$ rather than the absolute value of elements.

Note that the inputs to our reconstruction network are complex-value data, thus we transform every operator above to its complex value counterparts by following complex-value extension methods presented in \cite{Fu2019,2019Complex}.

\subsection{Network Training and Testing Details}

Here, we describe some details when training our proposed network, including training strategies and the testing process.

To train networks in a supervised way,
we first prepare $N_{tr} = 50,000
$ training samples with known labels, i.e., the ground truth $\bm x^*$, under \eqref{eq:system}. We also generate another $N_{vl} = 1000$ and $N_{ts} = 5000$ samples for validation and testing, respectively. The validation data set is used for determining some hyperparameters of the network. 
The initial values of ${\bm W_1}, {\bm W_2}$ (for \ac{adalista}), $\{{\bm W}_q\}_{q=1}^{Q}$ (for Ada-BlockLISTA) are chosen as identity matrices in the experiments. 

With these labeled training data and initialized parameters, we train the network implemented using TensorFlow with the strategy described below.

We choose \ac{nmse} of the recovered signal at the $T$-th  layer normalized by the norm of ground truth $\bm x^*$ as the loss function, which is defined as
\begin{equation}\label{eq:NMSE1}
\mathrm{NMSE} = 
{\rm E}{ 
	\frac{\norm{ \bm x^* - {\bm x}^{(T)} }_2}
	{\norm{ \bm x^* }_2}
}.
\end{equation}

Then we use Adam as optimizer \cite{kingma2014adam} 
and start optimizing the network parameters $\bm{\Theta}^{(t)} = \{{\bm W}_1, \cdots,{\bm W}_Q,\theta^{(t)}, {\gamma}^{(t)}\} $ towards minimizing the loss function on the training set with the learning rate initialed as $lr_0 = 0.0005$, which will be modified during training.
During the training process, we check the network performance and refine network hyperparameters by calculating the loss function on the validation data set.


Compared to classical ISTA and Block-ISTA which take nearly $100$ iterations or even more, \ac{adalista} and Ada-BlockLISTA networks improve the convergence speed by one order of magnitude or more, reconstructing signals within $10$ layers. 
By using the block information in both linear/nonlinear layers, our proposed network exhibits better block recovery performance against the number of non-zero blocks in recovered signals.


The testing process of Ada-BlockLISTA is detailed in Algorithm~\ref{alg:Algo1}.

\begin{algorithm}
	\caption{Ada-BlockLISTA Inference Algorithm}
	\label{alg:Algo1}
	\begin{algorithmic}[1]
		\STATE \underline{Input}: observation $\bm y$, dictionary $\bm \Phi$
		\STATE \underline{Initialization}: Set $t=0$, and generate an initial guess ${\bm x}_q^{(0)}= \bm 0$, $q=1,\cdots,Q$. 
		\STATE \label{stp:res} Compute the residual value ${\bm r^{(t)}} = \bm y - \sum\limits_{i = 1}^Q { {{\bm \Phi}_i}{\bm x}_i^{(t)}} $;
		\FOR {each $q \in [1,Q]$}
		\STATE \label{stp:IC} {\em Gradient descend:} refer to \eqref{eq:AdaBlockLISTA-sub1}
		\STATE \label{stp:SoftDec} {\em Block-wise soft thresholding:} 
		refer to \eqref{eq:AdaBlockLISTA-sub2}
		\ENDFOR
		\STATE Set $t := t+1$. If $t < T $ return to Step \ref{stp:res}.	
		\STATE  \underline{Output}: block-sparse signal $\bm x_T$.
	\end{algorithmic}
\end{algorithm}

\section{Convergence Analysis}
\label{sec:Convergence}


In this section, we state the main theoretical results of this paper, i.e., the linear convergence rate and support recovery guarantee of Ada-BlockLISTA.
First, we introduce some essential tools.
Throughout this section, $\bm x^*$ is a block-sparse signal, as in \eqref{eq:blockx}, comprised of $Q$ blocks of length $P$, where the number of non-zero blocks of $\bm x^*$ is $s$.

\noindent\textbf{Assumption 1}: The $\ell_2$ norm of $\bm x^*_q$ is bounded by $\zeta$, for all $q = 1, \cdots, Q.$
We construct the set,
\begin{equation}
\begin{aligned}
\mathcal{X} (\zeta, s) \coloneqq   \big\{  \bm v \in \mathbb{C}^M \mid & \| { \bm v_q} \|_2 \leq \zeta, \forall q = 1, \cdots, Q, \\
& \|  \bm v\|_{2,0} \leq s \big\},
\end{aligned}
\end{equation}
to be the set of signals which have the same block-sparsity structure as $\bm x^*$. 

\begin{proposition}\label{prop:noisebound}
	Let $\bm \varepsilon \in \mathbb{C}^N$ be a complex standard normal random variable. Then, with probability $1 - \delta$ for $\delta \in (0,1),$ 
	$$\|\bm \varepsilon\|_2 < \sigma,$$
	where $\sigma := \sigma(N, \delta)$ is a constant which depends on the dimensions of the noise vector and the chosen tightness on the bound.
\end{proposition}

Next, it is necessary to recall two basic concepts of coherence \cite{Block-Yonina}, i.e., the mutual coherence within a block, also referred to as sub-coherence, and block-coherence of the dictionary $\bm \Phi$, which describe its local and global properties, respectively.

\begin{definition}
	Let $\bm{A}, \bm{B}$ be two matrices with consistent dimensions such that $\bm a^H_{i} \bm b_i = 1$, where $\bm a_i$ and $\bm b_i$ are the $i$-th columns of $\bm A$ and $\bm B$, respectively. 
	The mutual coherence between $\bm A$ and $\bm B$ is defined as,
	\begin{equation}
	\mu(\bm A, \bm B) = \max_{i \neq j} | \bm a^H_{i} \bm b_j |.
	\end{equation}
\end{definition}

\begin{definition}
	The sub-coherence of the dictionary ${\bm \Phi}$ characterizes the coherence between different columns within each sub-matrix ${\bm \Phi}_q \in {\mathbb{C}^{N \times P}}, \quad q = 1,2, \cdots, Q$, with normalized columns, which is defined as,
	\begin{equation}\label{eq:sub_coherence}
	{\nu _ I}
	=\max_{1 \le q \le Q}\max_{i \ne j}\left| {\bm \phi}_{i,q}^H {\bm \phi}_{j,q} \right|, \quad
	{\bm \phi}_{i,q},{\bm \phi}_{j,q}\in {\bm \Phi}_q.
	\end{equation}
\end{definition}
The sub-coherence of ${\bm \Phi}$ takes the maximum over all the absolute value of the off-diagonal entries in each diagonal block in the Gram matrix ${\bm \Phi}^H{\bm \Phi}$ (the orange entries in Fig. \ref{fig:gram}). 

\begin{definition}
	The block-coherence of the dictionary ${\bm \Phi}$ with normalized columns characterizes the coherence between different blocks, which is defined as, 
	\begin{equation}\label{eq:block_coherence}
	{\mu _ B} = \max_{1 \le q \le Q} \max_{q', q \ne q'}{{1}\over{P}}\left\| {\bm \Phi}_q^H {\bm \Phi}_{q'} \right\|_s,
	\end{equation}
	where $\norm{\cdot}_s$ denotes the spectral norm.
\end{definition}
The block-coherence of ${\bm \Phi}$ depends on the maximum over the spectral norm of all the off-diagonal blocks in the Gram matrix ${\bm \Phi}^H{\bm \Phi}$  (the yellow entries in Fig. \ref{fig:gram}).

With the definitions above, we restate Theorem 3 of \cite{Block-Yonina}, which shows that the exact recovery condition holds universally under certain conditions on the sub-coherence and block-coherence of the dictionary $\bm \Phi$.

\begin{theorem}\label{thm:convergence_Yonina}
	Let ${\mu _ B}$ be the block-coherence and ${\nu _ I}$ the sub-coherence of the dictionary $\bm \Phi$. A sufficient condition to recover $\bm x^*$ through Block-ISTA is that
	\begin{equation}\label{eq:s_Yonina}
	sP < \frac{1}{2}\left( { {\mu} _ B} ^{ - 1} + P - (P - 1)\frac{{ {\nu} _ I}}{{ {\mu} _ B}} \right).		
	\end{equation}
\end{theorem}	

Here we generalize the definitions above and illustrate the linear convergence rate of our Ada-BlockLISTA in the following theorem.

\begin{theorem}[Convergence rate of Ada-BlockLISTA]\label{thm:convergence1}
	Suppose that we are given the matrix $\bm \Phi \in \mathbb{C}^{N \times M}$ that follows the structure defined in \eqref{eq:blockdict}.
	Let $\{{\bm x}^{(t)}\}_{t=1}^{\infty}$
	be generated by a trained Ada-BlockLISTA with parameters $\{{\bm W}_1, \cdots,{\bm W}_Q,\theta^{(t)}, {\gamma}^{(t)}\}_{t=0}^{\infty}$ in \eqref{eq:AdaBlockLISTA} 
	and initialized as ${\bm x}^{(0)}= \bm 0$.
	
	Consider the three quantities,
	\begin{align}
	\label{eq:g_sub_coherence}
	\tilde{\nu }_ I &= 
	\underset{t\geq 0}{\max}~ 
	\max_{1 \le q \le Q}
	\Big\{ 
	\max _{i \neq j \atop 1 \leq i, j \leq P} \left| {\gamma}^{(t)} {\bm \phi}_{i,q}^H \bm{W}_q {\bm \phi}_{j,q} \right| 
	\Big\}, \\ 
	\label{eq:g_block_coherence}
	\tilde{\mu }_ B &= 
	\underset{t\geq 0}{\max}~ 
	\max_{1 \le i \le Q}
	\Big\{ 
	\max_{j, i \ne j}{{1}\over{P}}\left\| {\gamma}^{(t)} {\bm \Phi}_i^H \bm{W}_i {\bm \Phi}_j \right\|_s
	\Big\},  
	\\
	\label{eq:C_W}
	C_W &= 
	\underset{t\geq 0}{\max}~ 
	\max_{1 \le q \le Q}
	\| {\gamma}^{(t)} {\bm \Phi}_q^H \bm{W}_q \|_{2,1}.
	\end{align}
	
	Then, suppose $s$ to be sufficiently small such that
	\begin{equation}\label{eq:s}
	s < \frac{1}{{2P}}\left( { \tilde{\mu} _ B} ^{ - 1} + P - (P - 1)\frac{{ \tilde{\nu} _ I}}{{ \tilde{\mu} _ B}} \right).
	\end{equation}
	Consider the system model \eqref{eq:system};
	let $\bm x^* \in \mathcal{X}(\zeta, s)$, and $\bm \varepsilon$ represent additive random noise and $\sigma$ be the high-probability upper bound from Proposition \ref{prop:noisebound}.
	Assume that the thresholds $\theta ^{(t)}$ are
	\begin{equation}\label{eq:theta}
	\theta ^{(t)} = P\tilde{\mu }_ B \sup\limits_{\bm{x}^* \in \mathbb{C}^M} \| {\bm x}^{(t)} - {\bm x}^* \|_{2,1}+ C_{\bm W} \sigma. 
	\end{equation}
	Then, with probability $1 - \delta$ for $\delta \in (0,1),$ we have: 
	\begin{enumerate}
		\item The support of the recovered signal is contained in the true support, i.e. $\Supp({\bm x}^{(t)}) \subseteq \Supp (\bm x^*)$.
		\item The recovered error bound is
		\begin{equation}\label{eq:induction_hyp_error}
		\| {\bm x}^{(t)} - \bm x^* \|_{2,1} \leq s \zeta \exp (-c_1 t)+ c_2 \sigma,
		\end{equation}
		where $c_1>0$ and $c_2>0$ are constants that depend on $\bm{ \Phi }$ and $s$.
	\end{enumerate}
\end{theorem}
The proof of Theorem \ref{thm:convergence1} mimics the corresponding proof steps in \cite{liu2018alista, 2020Ada, Zarka2020Deep} and can be viewed as an extension to the block-sparse case, which is detailed in the Appendix.

The above theorem can be interpreted in two aspects:
\begin{itemize}
	\item
	The conclusion of Theorem \ref{thm:convergence1} gives an upper bound of the recovered error at the $t$-the layer of Ada-BlockLISTA, which reduces to $	\| {\bm x}^{(t)} - \bm x^* \|_{2,1} \leq s \zeta \exp (-c_1 t)$ in noiseless case. Thus, Theorem \ref{thm:convergence1} shows that under sufficient conditions \eqref{eq:s} and \eqref{eq:theta} Ada-BlockLISTA converges at a $ \mathcal{O}(\log(\frac{1}{\epsilon}))$ rate, which is faster than original ISTA of $\mathcal{O}(\frac{1}{\epsilon})$ and FISTA of $\mathcal{O}(\frac{1}{\sqrt{\epsilon}})$.
	\item
	Compared to the recovery condition of \ac{adalista} \cite[Theorem 1]{2020Ada}, the recovery condition \eqref{eq:s}, which is the same with the block-sparse recovery condition of \cite[Theorem 3]{Block-Yonina}, show that exploring the block structure in $\bm x$ leads to ensure recovery of $\bm x$ with higher sparsity level.
	Note that \eqref{eq:s} is a sufficient condition to guarantee successful recovery, 
	which may not be necessary for practice.
\end{itemize}
Table~\ref{table1} compares the convergence rate and recovery condition of \ac{adalista} and Ada-BlockLISTA networks as well as corresponding traditional iterative algorithms.

\begin{table} 
	
	\begin{tabu} to 1\columnwidth{X[3,c]|X[4,c]|X[5,c] |X[7,c]} 
		\hline 
		ISTA & \ac{adalista} & Block-ISTA & \textbf{Ada-BlockLISTA} \\ 
		\hline 
		$\mathcal{O}(1/{\epsilon})$ &
		$\mathcal{O}(\log(1/{\epsilon}))$ &
		$\mathcal{O}(1/{\epsilon})$ &
		$\mathcal{O}(\log(1/{\epsilon}))$ \\ 
		\hline 
		\multicolumn{2}{c|} {$s < \frac{1}{2}\left( { \overline{\mu} } ^{ - 1} + 1 \right)$} & 
		\multicolumn{2}{c} {$s < \frac{1}{{2P}}\left( { \overline{\mu} _ B} ^{ - 1} + P - (P - 1)\frac{{ \overline{\nu} _ I}}{{ \overline{\mu} _ B}} \right)$} \\ 
		\hline 
	\end{tabu} 
	\caption{The convergence rate (first row) and recovery condition (second row) of ISTA, Block-ISTA, \ac{adalista} and Ada-BlockLISTA.
		Here we define $\overline{\mu} = \mu(\bm \Phi, \bm \Phi) $ (for ISTA) or $\overline{\mu} = \mu(\bm W_2 \bm \Phi, \bm \Phi) $ (for \ac{adalista}). 
		Correspondingly, for block sparse recovery, the sub-coherence $\overline{\nu} _ I$ is ${\nu} _ I$ in \eqref{eq:sub_coherence} or $\tilde{\nu} _ I$ in \eqref{eq:g_sub_coherence} and the block coherence $\overline{\mu} _ B$ is ${\mu} _ B$ in \eqref{eq:block_coherence} (for Block-ISTA) or $\tilde{\mu} _ B$ in \eqref{eq:g_block_coherence} (for Ada-BlockLISTA).
	} \label{table1}
\end{table} 

\section{Numerical Results}
\label{sec:Numerical}

Here, we consider an application of the block-sparse signal model known as range-Doppler estimation of extended targets, which is similar to \ac{2d} harmonic retrieval. 
We apply each of the algorithms previously defined to this application.

\subsection{Signal model of range-Doppler estimation problem} \label{subsec:2dMHR}
Frequency agile radars are very attractive for tasks under complex electromagnetic environments \cite{ExtendedTarget, 4137843,4338057,6212202}. 
Following the signal model presented in \cite{ExtendedTarget}, we first derive the echo expressions of frequency agile radars, which have a high synthetic range resolution.
Next, we show that range-Doppler estimation of extended targets can be formed into a block-sparse recovery problem.

The radar uses $N$ pulses whose carrier frequency is computed by $f_n = {f_0} + {C_n} {\Delta f}$, $n = 0,1,2,\cdots, N-1$, 
where $f_0$ is the initial carrier frequency, and $\Delta f$ is the frequency step interval. 
The randomized modulation code $C_n$ is randomly selected from an integer set \{$0, 1, 2, \cdots, P-1$\}, where $P$ is the number of frequency points. 
The $n$-th transmitted pulse, $s(n,t)$, can be expressed as
\begin{equation}
s(n,t) = \text{rect}\left(\frac{t - nT_r}{T_p} \right) \exp\{j2\pi {f_n}\left( {t - n{T_r}} \right)\},
\end{equation}
where $T_r$ is the pulse repetition interval (PRI), 
$ T_p $ is the pulse width and $\text{rect}(\cdot)$ is the rectangular function defined as
\begin{equation}
\text{rect}(t) = \left\lbrace 
\begin{array}{ll}
1 & 0 \leq t \leq 1,\\ 
0 & \text{otherwise}. 
\end{array}\right.
\end{equation}
We first consider the received signal of a single ideal scatterer with complex scattering coefficient, $\beta$.
Assuming that the scatterer's range is $R$ and velocity is $v$, its time delay is $ {\tau }(t) = \frac{2({R} + v t)}{c}$, where $c$ is the speed of light. 
Based on the “stop-and-go” assumption \cite{Richardsbook}, the $n$-th echo-signal $s_{\rm r}(n,t)$ from the scatterer is written as 
\begin{equation}
{s_{\rm r}}(n,t) = \beta 
\text{rect}\left(\frac{ t \!\!- \!\!n{T_r} \!\!-\!\! {\tau }(n{T_r}) }{ T_p } \right)
\exp\{ j2\pi {f_n}\left( {t \!\! - \!\!n{T_r} \!\!-\!\! {\tau}(n{T_r})} \right)\}.
\end{equation}
Then, the echo of each pulse is down-converted by its corresponding carrier frequency. 
The baseband echo, ${\widetilde s_{\rm r}}(n,t)$, becomes
\begin{align}
{\widetilde s_{\rm r}}(n,t)
&= {s_{\rm r}}(n,t) \cdot e^{ - j2\pi f_n \left(t - n{T_r}\right)} \notag\\
&= \beta 
\text{rect}\left(\frac{t \!\!- \!\!n{T_r} \!\!-\!\! {\tau }(n{T_r})}{T_p} \right)
\exp\{  - j2\pi {f_n} {\tau}(n{T_r}) \}. 
\end{align}

Then, the baseband echo signal ${\widetilde s_{\rm r}}(n,t)$ is sampled at time instant $t = n{T_r}+{t_s}$, where $t_s = l_s/f_s$ and $l_s=0,1,…,\lfloor T_r f_s \rfloor$.
Every sample time instant corresponds to a coarse range cell (CRC) \cite{9205659}.
Assuming that the scatterer does not move between CRCs during a coherent processing interval (CPI), i.e., there exists an integer $l'_s$ such that $ (l'_s-1)/f_s \le 2R/c \le l'_s/f_s$, the echo sequence $y(n)$ of the scatterer sampled at $t = n{T_r}+{l'_s/f_s}$ is defined as (Please refer to \cite{ExtendedTarget} for detailed assumptions and derivations)
\begin{align}
y(n) 
&= {\widetilde s_{\rm r}}(n,t)|_{t = n{T_r} + l'_s/f_s} \notag \\
&= \beta \exp\{ - j2\pi {f_n}{\tau}(n{T_r})\} \notag \\
&= \beta \exp\left\{ - j\frac{{4\pi }}{c}({f_0} + {C_n}\Delta f)\left( R + v n{T_r} \right)\right\}.
\label{eq:sampled}
\end{align}

\if false
Then, the echo sequence can be represented as 
\begin{equation}\label{eq:rsfr}
y(n) = \sum\limits_{k = 1}^K {{\widetilde \beta _k}\exp\left\{j\frac{{4\pi }}{c}({f_0} + {C_n}df)\left( {{R_k} + n{T_r} {v_k}} \right)\right\} } .
\end{equation}
\fi

Due to the high synthetic range resolution of the frequency agile radar, radar targets could span multiple high range resolution (HRR) cells and are referred to as \emph{extended targets} \cite{ExtendedTarget}. We assume that the $k$-th target is extended in range, 
which consists of $P_k$ scatterers with the same velocity $v_k$, corresponding to $P_k$ HRR cells. 
When there are $K$ targets, we rewrite \eqref{eq:sampled} as 
\begin{equation}\label{eq:rsfr}
\!\!\! y(n) = \sum _{k=1}^{K} \sum _{i=1}^{P_k}
{ \beta _{k,i} \exp\left\{ - j\frac{{4\pi }}{c}({f_0} + {C_n} \Delta f)\left( { R_{k,i} + {v_k} n{T_r} } \right)\right\} },
\end{equation}
where $\beta _{k,i}$ and  $R_{k,i}$ are the scattering coefficient and range of $i$-th scatterer in the $k$-th target, respectively.

To further build a block-sparse recovery model, we discretize the range and velocity space into $P$ and $Q$ grids respectively, yielding $P \times Q$ discrete grid points, i.e., $\{R_p\}_{p \in \mathcal{P}} \times \{ v_q\}_{q \in \mathcal{Q}}$. 
As the maximum unambiguous range is $\frac{c}{2 \Delta f }$, we set the range grids as $R_p = \frac{c}{2 \Delta f } \frac{p}{P}$, $p \in \mathcal{P}$.
These $P$ range grids of the same velocity grid ($v_q$) form a velocity block, representing the HRR profile of a certain extended target moving at $v_q$, where $v_q = \frac{2c}{f_0 T_r} \frac{q}{Q}$, $q \in \mathcal{Q}$. 
Assuming that each scatterer's range and velocity lies on the prescribed grid points, 
we use $\bm{x}^*\in \mathbb{C}^{P Q} $ as the block-sparse signal to recover, which represents scatterers of targets, given by
\begin{equation}
\label{eq:blockx_radar}
\bm{x}^* \!=\! [
{\bm x}_1^T \ 
{\bm x}_2^T \ 
\cdots \ 
{\bm x}_Q^T
]^T,
\end{equation}
where the $q$-th velocity block is ${\bm x}_q = [
{\Gamma}_{q,1} \ 
{\Gamma}_{q,2} \ 
\cdots \ 
{\Gamma}_{q,P}
]^T \in \mathbb{C}^{P} $ and its elements correspond to the scattering coefficients of the extended target, given by
\begin{align*} 
{\Gamma}_{q,p}:= 
\left\lbrace 
\begin{array}{cc} 
\beta _{k,i} & \text{if }\, 
\exists (k, i),
\left( R_{k,i}, v_k\right)
=\left(R_{p}, v_q\right)\\ 
0 & \text{otherwise}. 
\end{array}\right. 
\tag{10} 
\end{align*}
Here, $ \bm x^*  $ is partitioned into $Q$ velocity blocks and each block has $P$ elements with the same velocity.
If there are $K$ targets, $ \bm x^*$ has $K$ non-zero blocks. The block corresponding to velocity $v_k$ contains $P_k$ non-zero elements, according to \eqref{eq:rsfr}.
Due to the natural block structure,  $ \bm x^*$ can be modeled as a block-sparse vector when there are only a few non-zero blocks, i.e., $K \ll Q$.

Accordingly, the measurement matrix can be written as
\begin{equation}
\bm \Phi  = \left[
\bm \Phi \left( {{v_1}} \right), 
\bm \Phi \left( {{v_2}} \right), \cdots , 
\bm \Phi \left( {{v_Q}} \right) 
\right],
\end{equation}
where the $k$-th velocity block $\bm \Phi \left( {{v_k}} \right)$ is a sub-matrix of size $N \times P$, whose $n$-th element of the $i$-th column is
${\phi _n}\left(  { R_{k,i},{v_k} } \right) = \exp\{  -  j\frac{{4\pi }}{c}({f_0} + {C_n}\Delta f)\left( { R_{k,i} + v_k n{T_r} } \right)\} $, $n = 0,1,2,\cdots,N-1$ and $i = 1,2,\cdots,P$.

Thus, we arrive at the block-sparse observation model,
\begin{equation} 
\bm y = \bm \Phi \bm x^* + \sigma_w {\bm w},
\end{equation}
where $\bm w \in \mathbb{C}^N$ is defined as a complex standard normal random vector and $\sigma_w > 0$ is the standard deviation of the additive noise.
When the measurement matrix is viewed as a row-sampled 2D \ac{dft} matrix, the problem can be viewed as a compressive \ac{2d} harmonic retrieval problem \cite{2foldblockToep}. 

\if false
of a Kronecker product, denoted by $\otimes$, of two discrete Fourier matrices, given by

\begin{equation} \label{eq:fulldict}
\begin{aligned}
\bm \Psi = 
\bm F_{P} \otimes \bm F_{Q}  \in \mathbb{C}^{M \times M}.
\end{aligned}
\end{equation}
The $(i, j)$-th entries of  $ \bm F_{P}$ and $\bm F_{Q}$ are defined as
$[\bm F_{P}]_{i,j}  = e^{\mathrm{j} 2 \pi \frac{ij}{P} }, i,j\in \mathcal{P}$,
and $[\bm F_{Q}]_{i,j}  = e^{\mathrm{j} 2 \pi \frac{ij}{Q} }, i,j\in \mathcal{Q}$, respectively.

The full observation vector $\bm y^{full} \in \mathbb{C}^{PQ} $ can be expressed as \begin{eqnarray} \label{eq:MD_dict}
\bm y^{full}  = \bm \Psi \bm x^*.
\end{eqnarray}

Furthermore, we consider a compressive measurement, where only $N$ entries of $\bm y^{full}$ are observed with $N \ll M$. 
To store the indices of selected entries from $\bm y^{full}$, we define a subset $\Omega$ of cardinality $N$ randomly chosen from the set $\mathcal{M}$.
Then we use a row-subsampled matrix $\bm R$ to selects $N$ rows of $\bm \Psi$ corresponding to the elements in $\Omega$, i.e., $\left[ \bm R \right]_{n,m} = 1$, where $m$ is the $n$-th element of $\Omega$ while other entries in the $n$-th row being zeros. 
Thus, the sub-sampled observations with AWGN $\bm \varepsilon$ are denoted by $\bm y \in \mathbb{C}^{N}$, i.e.,

\begin{equation}
\bm y = \bm R \bm y^{full}  + \bm \varepsilon
= \bm R \bm \Psi \bm x^*  + \bm \varepsilon,
\label{eq:partial_model}.
\end{equation}
Here, we use $ {\bm \Phi} = \bm R \bm \Psi$ to represent the new dictionary matrix, consisting of sub-sampled $N$ rows of the full dictionary $\bm \Psi$.
\fi
We aim to accurately recover all the positions of non-zero blocks in $\bm x$ by leveraging its block-sparse structure.
Experimental analysis of different block-sparse recovery methods and their non-block counterparts are provided in Subsection~\ref{ssec:results}.

\subsection{Simulation Results}
\label{ssec:results}

We simulate various experiments with different signal dimensions, as well as the number of non-zero blocks, in both the noiseless and noisy cases. 
In both cases, we compare the convergence and recovery accuracy of four methods (ISTA, Block-ISTA, \ac{adalista} and our Ada-BlockLISTA).

\subsubsection{Noiseless block sparse recovery}
\label{ssec:Noiseless}

In our simulation, the range and velocity frequencies are divided into $P = 16$ and $Q = 64$ grids, respectively. We generate observed signals according to the \eqref{eq:rsfr}, where the number of non-zero blocks in $\bm x$ is $s$.

We first show the improvement of convergence speed of deep unfolded networks, i.e., \ac{adalista} and our Ada-BlockLISTA).
With block sparsity $K=1$, the recovered NMSE results for each iteration (ISTA and Block-ISTA) or layer (\ac{adalista} and Ada-BlockLISTA) are shown in Fig.~\ref{fig:nmse_periter_noiseless}. 
Compared to classical ISTA and Block-ISTA which take nearly hundreds of iterations to converge, \ac{adalista} and Ada-BlockLISTA networks improve the convergence speed by one order of magnitude or more, which takes only about $5$ layers.

\begin{figure} [tb]
	\centering
	\includegraphics[width=0.4\columnwidth]{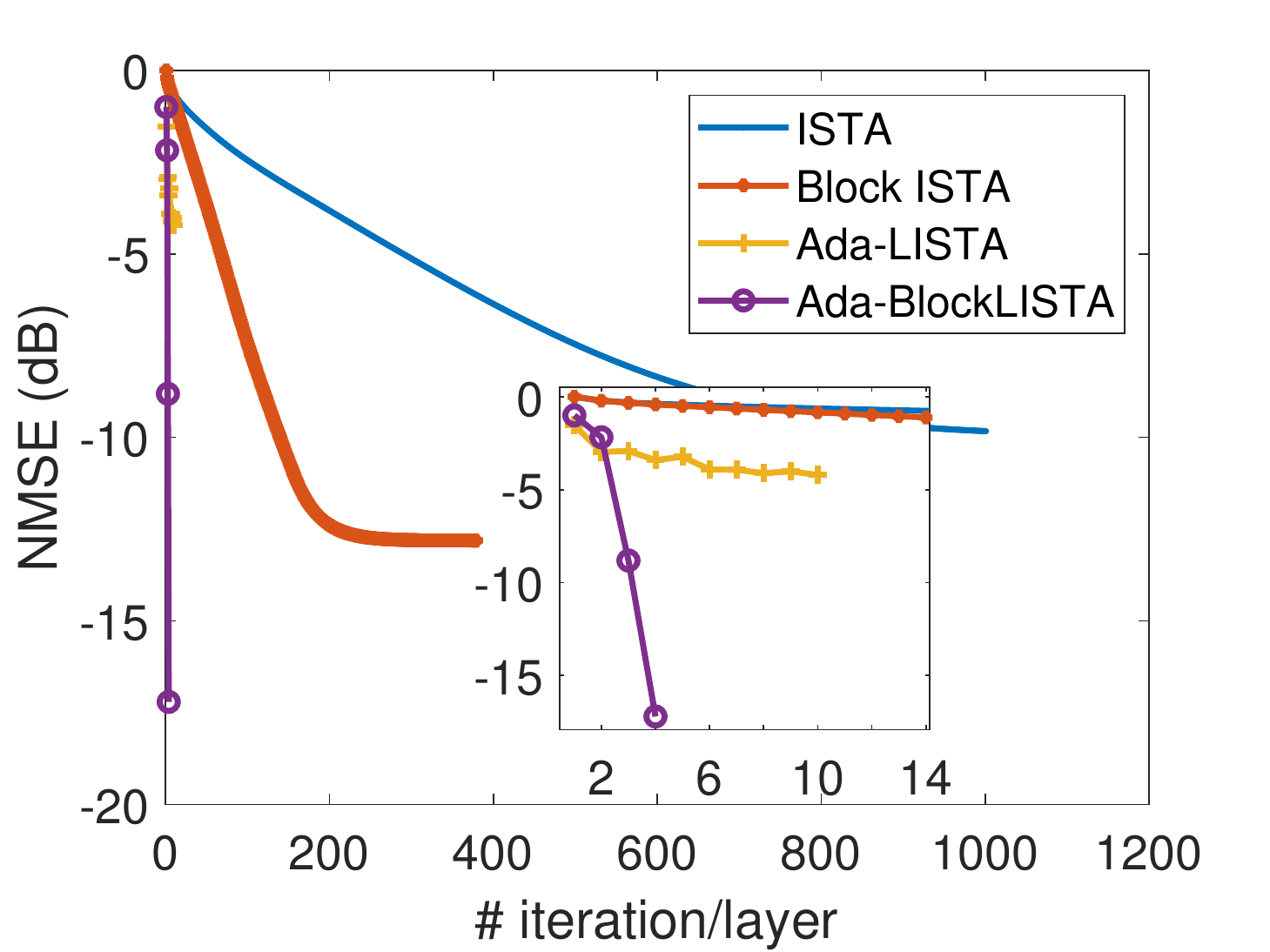}
	\caption{The NMSE of four methods in each iteration/layer without noise.}
	\label{fig:nmse_periter_noiseless}
\end{figure}

Furthermore, we also compare the recovery performance of these four methods with different block sparsity.
When there is only one none-zero block in $\bm x$, i.e., block sparsity $K=1$, the recovered block sparse signal  $\bm x$ of different methods are plotted in Fig.~\ref{fig:noiseless_oneblock}. 

\begin{figure}
	\centering 
	\subfigure[]{ 
		\includegraphics[width=0.2\columnwidth]{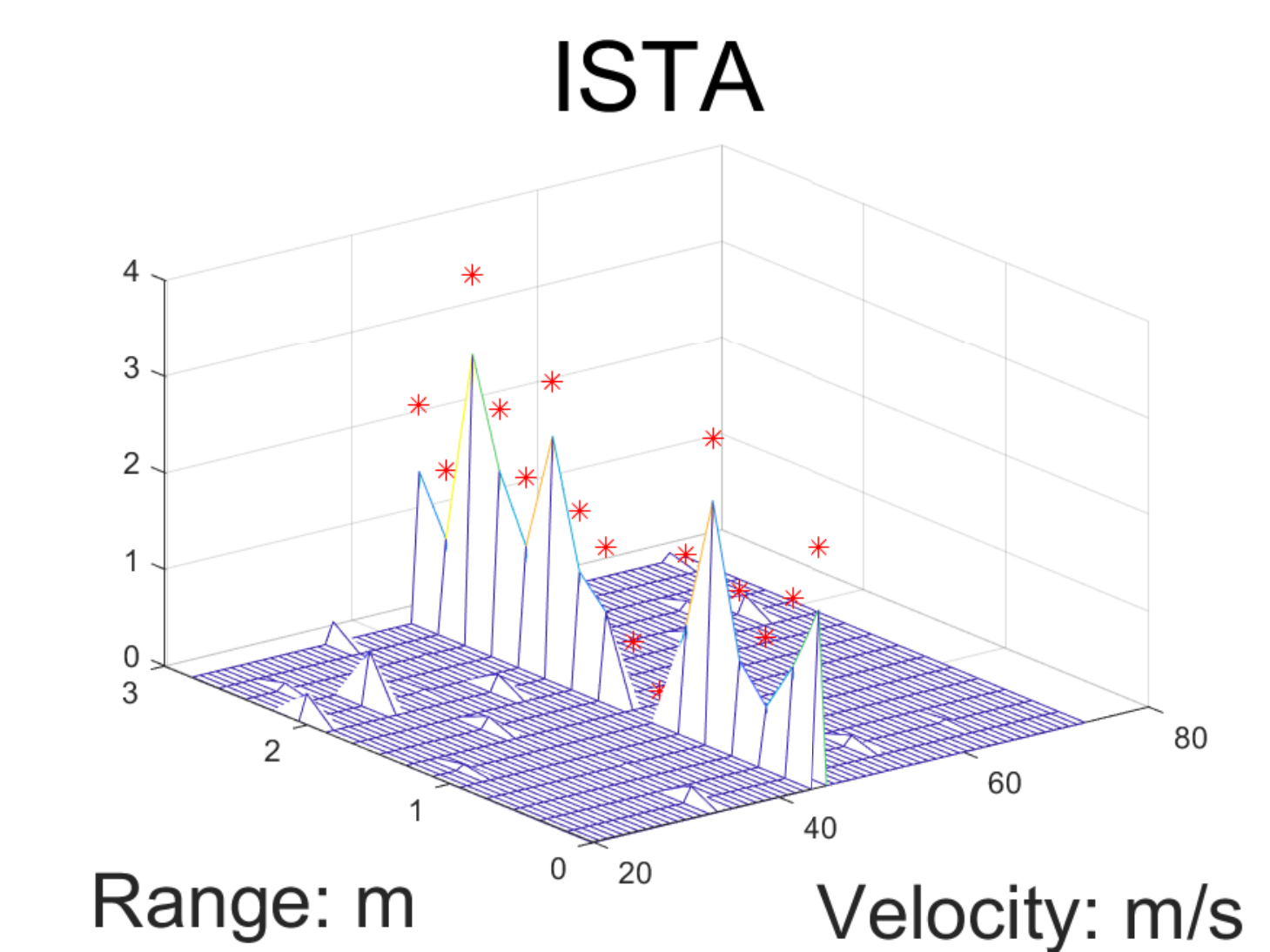}
	} 
	\subfigure[]{
		\includegraphics[width=0.2\columnwidth]{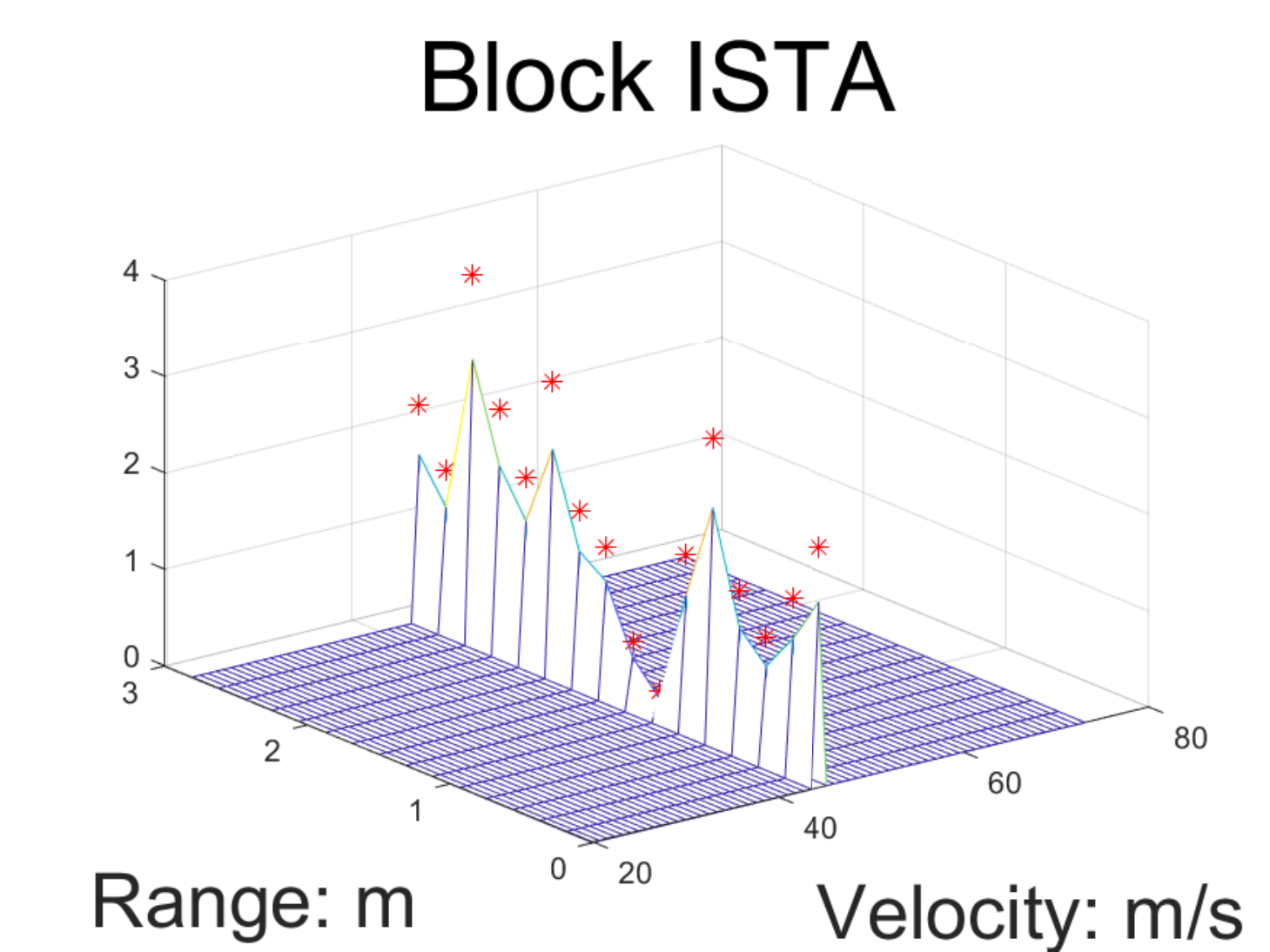}
	} 
	\subfigure[]{
		\includegraphics[width=0.2\columnwidth]{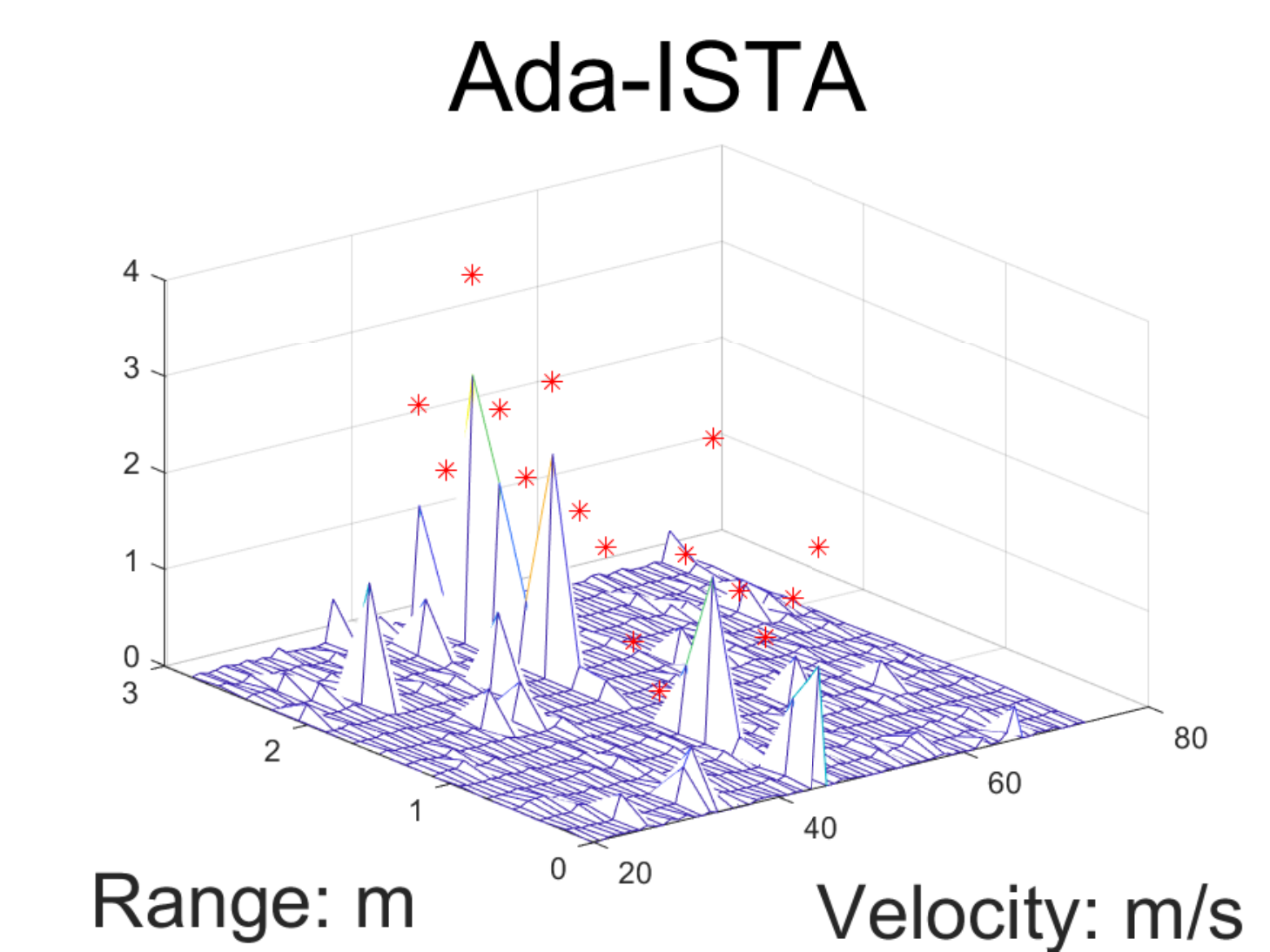}
	} 
	\subfigure[]{
		\includegraphics[width=0.2\columnwidth]{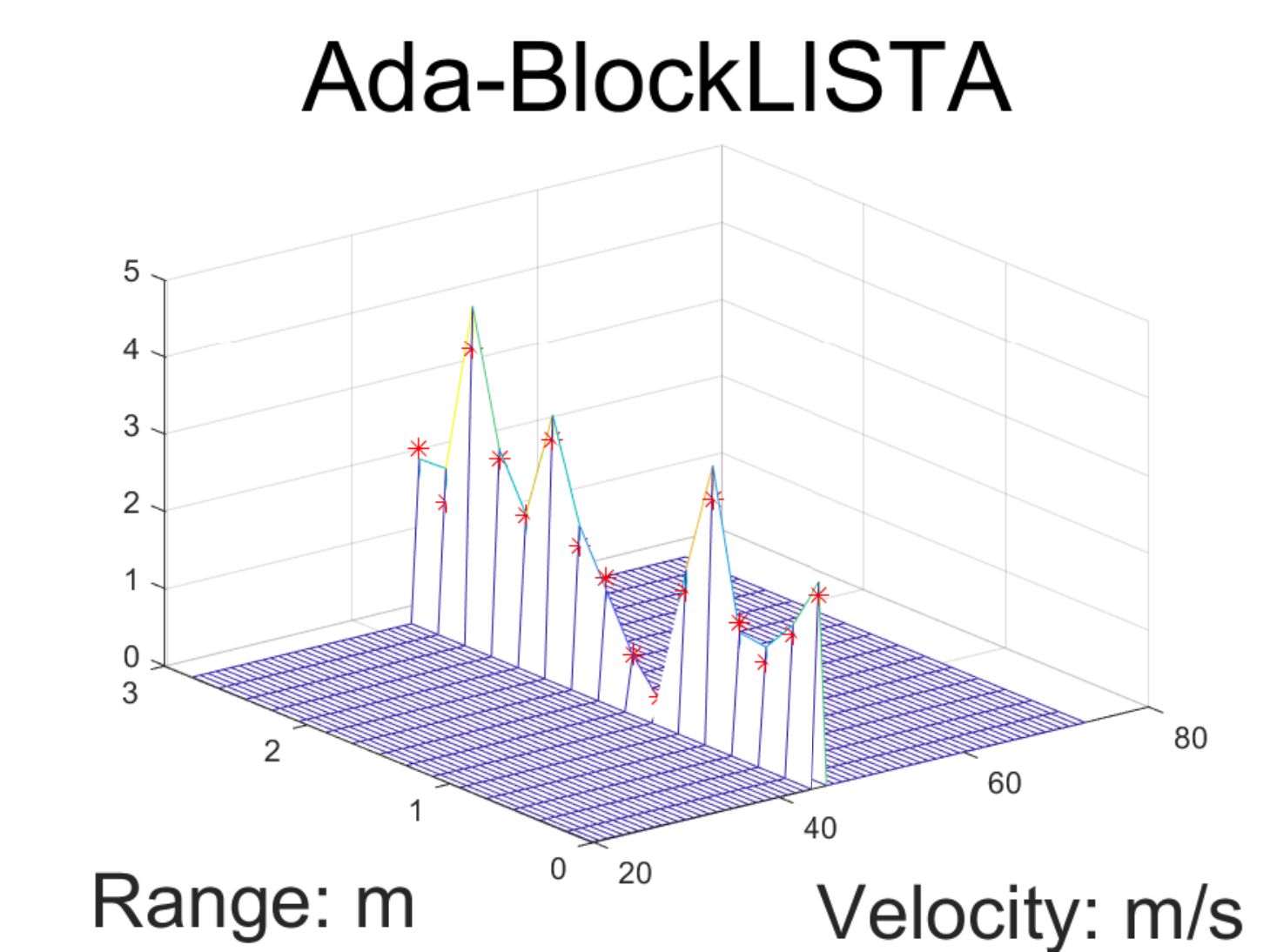}
	} 
	\caption{The reconstructed block sparse signal using (a) ISTA, (b) Block ISTA, (c) Ada-LISTA and (d) Ada-BlockLISTA with only one non-zero block.}
	\label{fig:noiseless_oneblock}
\end{figure}

Here we can see that the recovered plane of ISTA and \ac{adalista} both have a side-lobe pedestal, although there is no noise in observation. On the other hand, block sparse recovery methods such as Block ISTA and Ada-BlockLISTA have good recovery results.

A higher block sparsity $K$ is also simulated to further indicate the robustness of block sparse recovery methods.
The following Fig.~\ref{fig:noiseless_twoblock} shows the recovered block sparse signals With block sparsity $K=2$.
The results show that Block ISTA and Ada-BlockLISTA successfully reconstruct two blocks, while ISTA and \ac{adalista} fail.
It demonstrates that block sparse recovery methods can recover a larger number of blocks than their non-block counterparts, which is also proved in Theorem \ref{thm:convergence1}.

\begin{figure}
	\centering 
	\subfigure[]{ 
		\includegraphics[width=0.2\columnwidth]{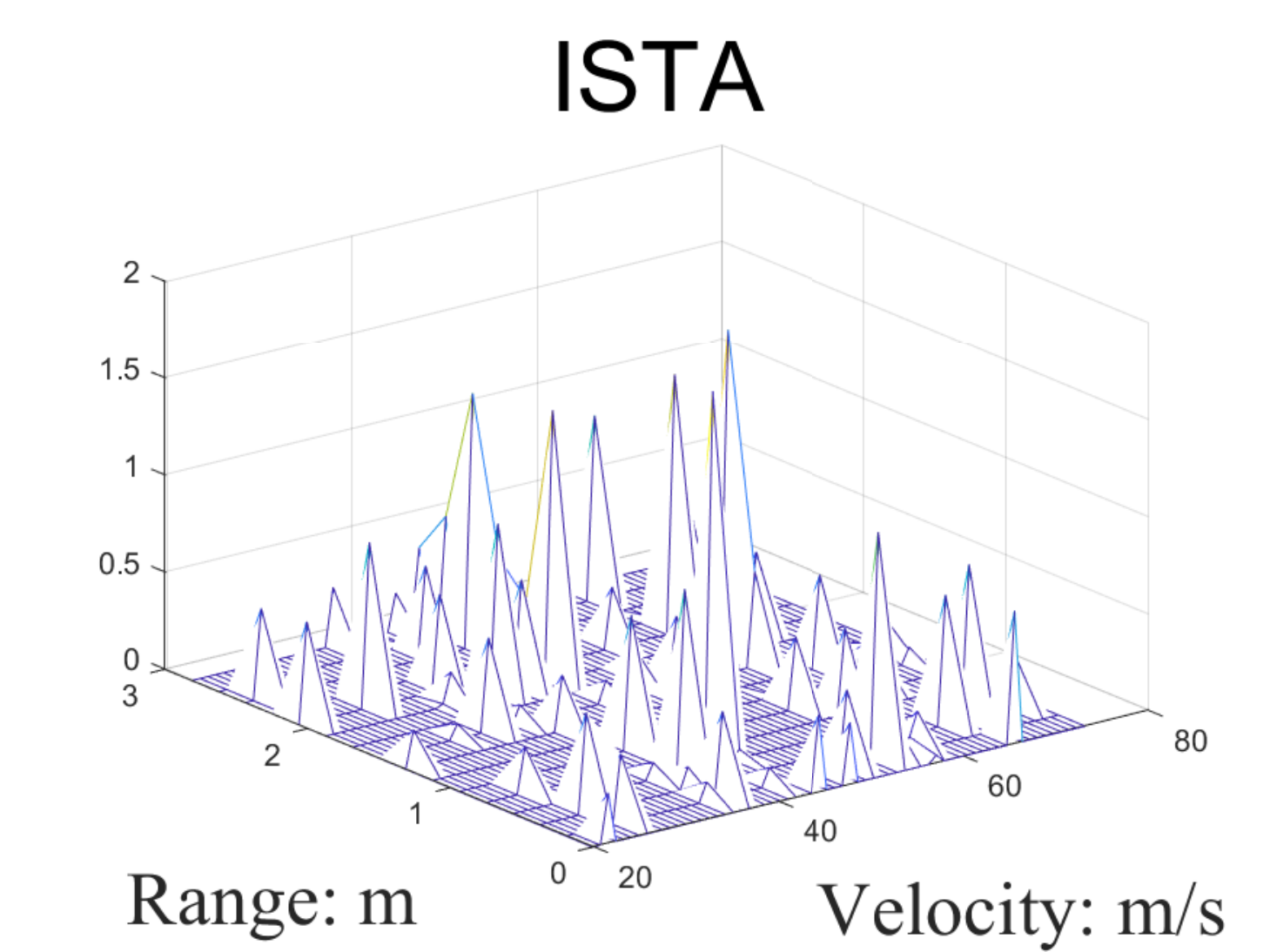}
	} 
	\subfigure[]{
		\includegraphics[width=0.2\columnwidth]{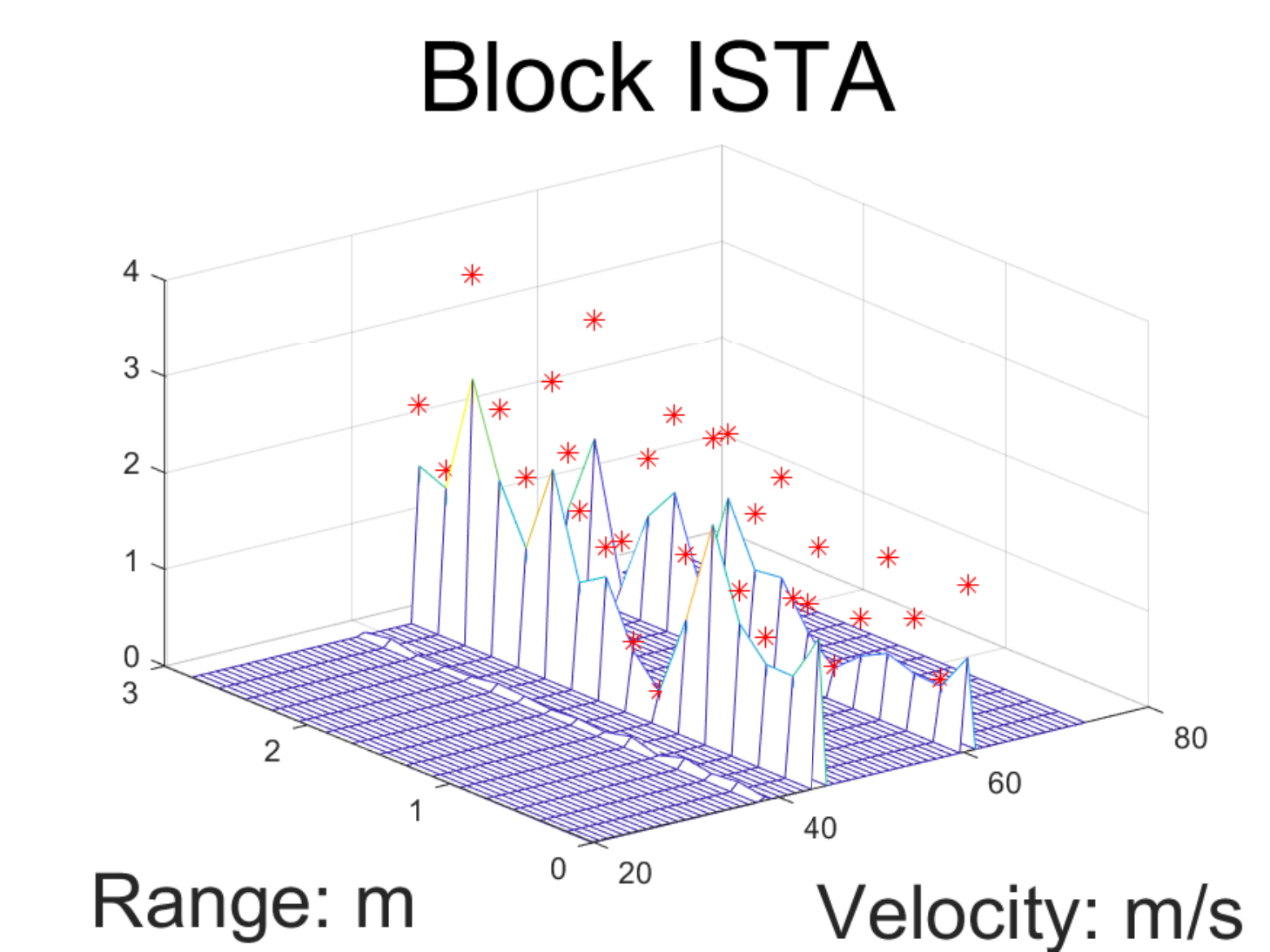}
	} 
	\subfigure[]{
		\includegraphics[width=0.2\columnwidth]{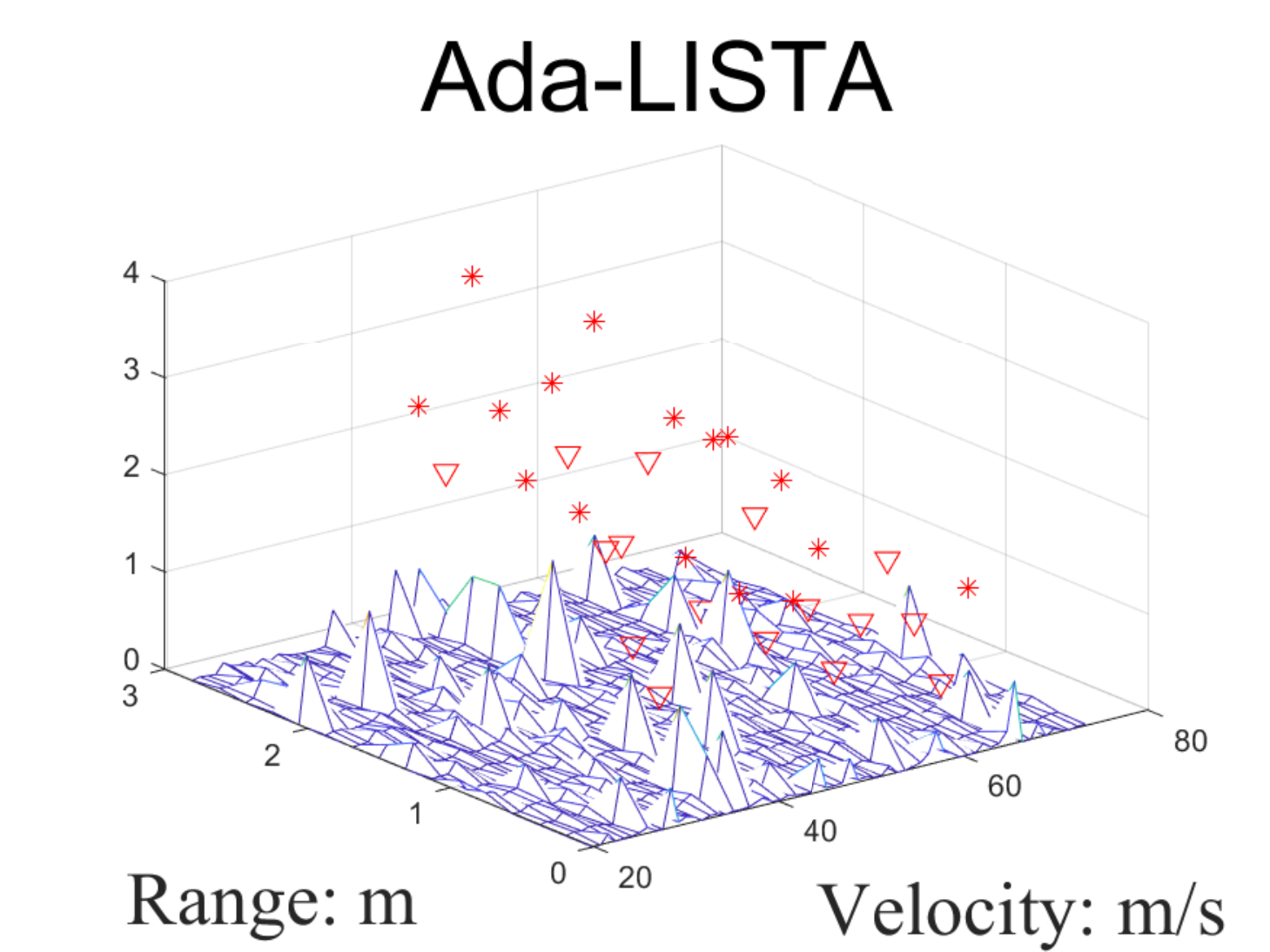}
	} 
	\subfigure[]{
		\includegraphics[width=0.2\columnwidth]{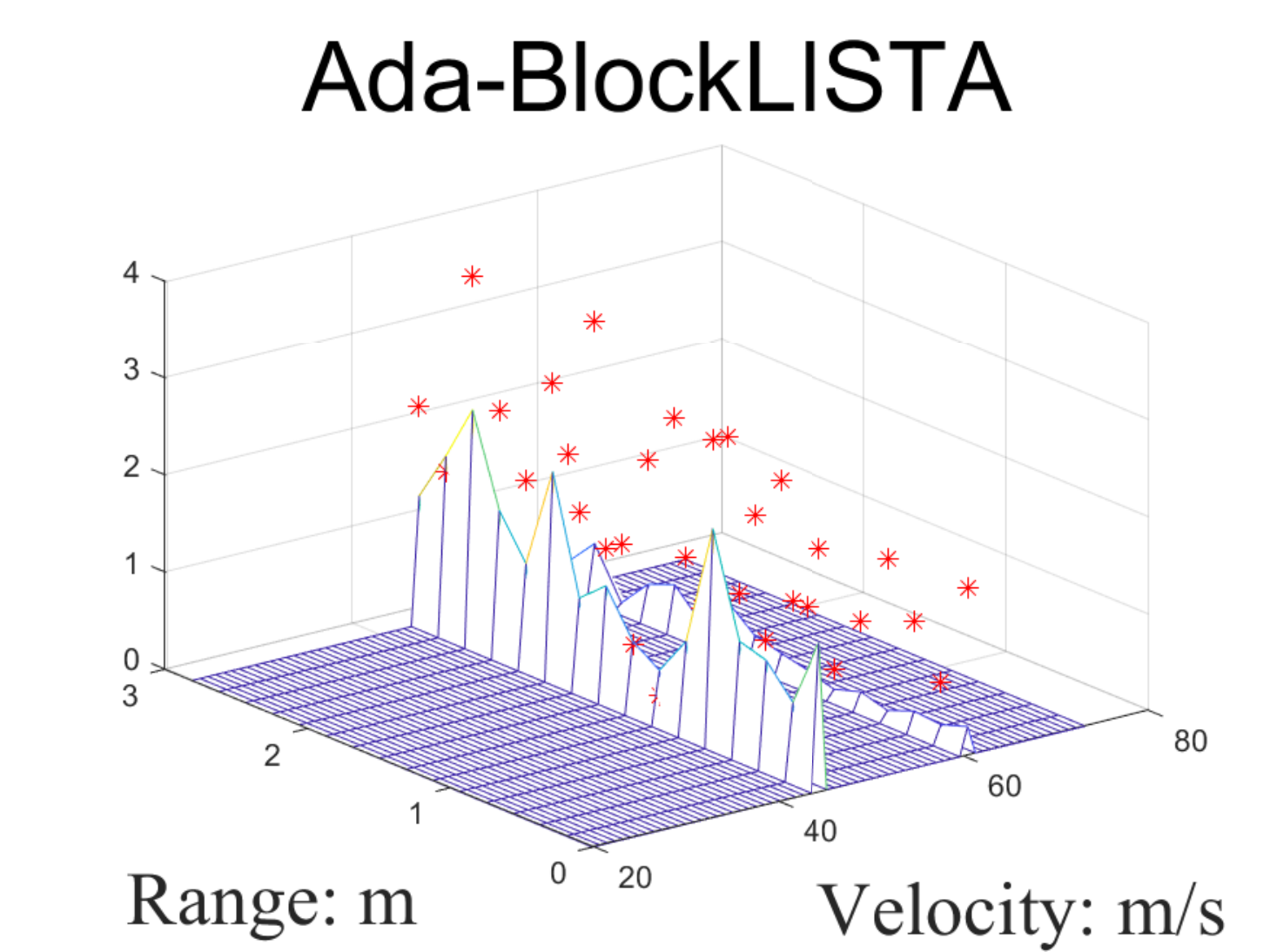}
	} 
	\caption{The reconstructed block sparse signal, which has two non-zero blocks, using (a) ISTA, (b) Block ISTA, (c) Ada-LISTA and (d) Ada-BlockLISTA.}
	\label{fig:noiseless_twoblock}
\end{figure}

\subsubsection{Noisy block sparse recovery}
\label{ssec:Noisy}

In this subsection, we analyze the simulation results of four methods (ISTA, Block-ISTA, \ac{adalista} and our Ada-BlockLISTA) with noise.
Here, we let $P = 4$ and $Q = 64$.

Fig.~\ref{fig:nmse_periter_noisy} shows that \ac{adalista} and Ada-BlockLISTA both enjoy accelerated convergence speed over their respective traditional iterative algorithm. 
While Block-ISTA improves upon ISTA with higher reconstruction quality, our Ada-BlockLISTA network shows both rapid convergence speed and great recovery performance.

\begin{figure} [tb]
	\centering
	\includegraphics[width=0.4\columnwidth]{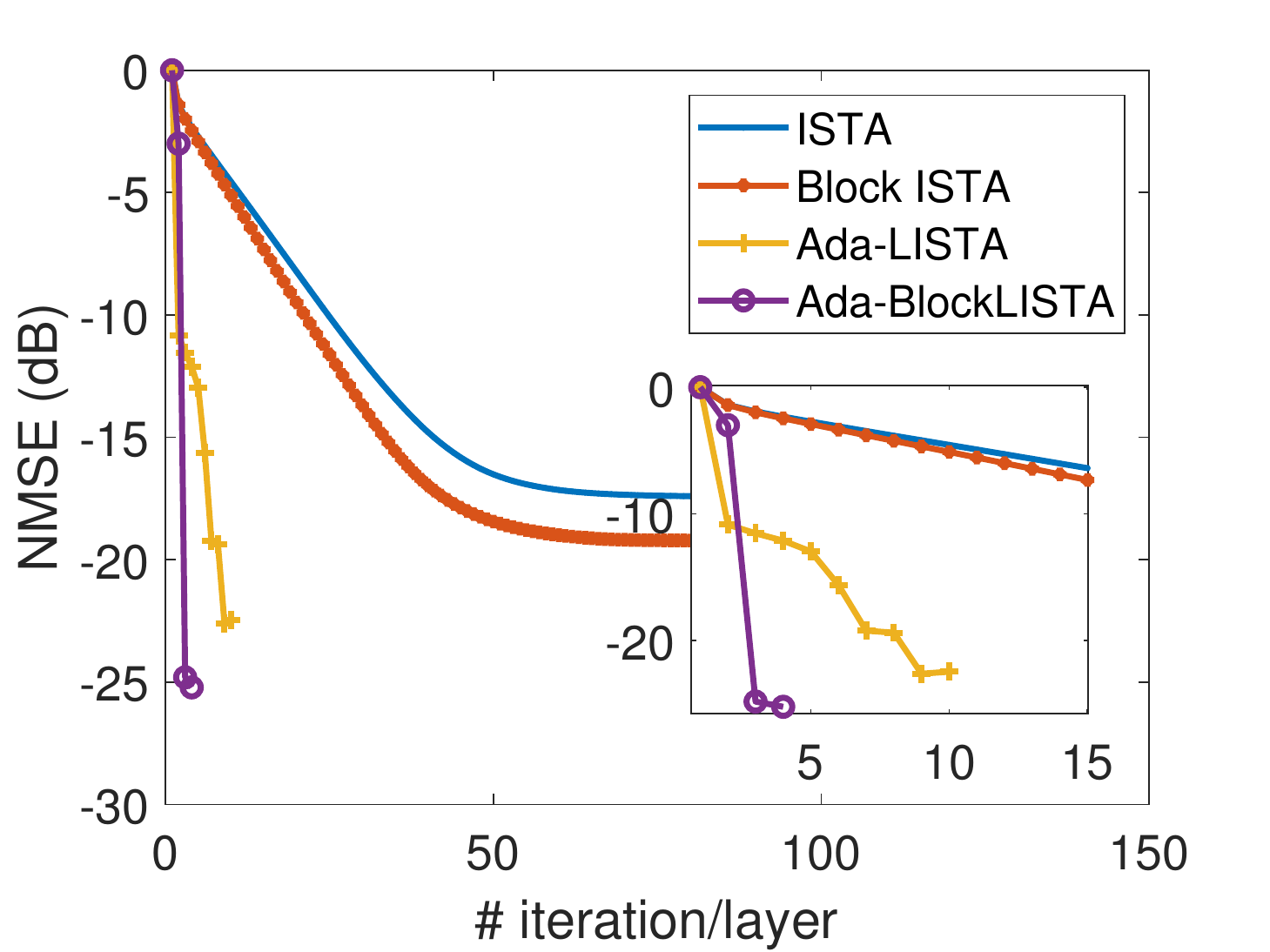}
	\caption{The NMSE of four methods in each iteration/layer with noise.}
	\label{fig:nmse_periter_noisy}
\end{figure}

We also evaluate the support recovery performance of block-sparse signals in terms of hit rate versus \ac{snr} and block sparsity in Fig.~\ref{fig:hit rate}.
The \ac{snr} is computed as $\mathrm{SNR} = 10\log _{10} \frac{1}{
	\sigma_w^2 }$ , where $\sigma_w^2$ is the noise variance.
The hit rate is defined as the percentage of successes in finding the nonzero entries in $\bm x$ against noise.
In Fig.~\ref{fig:hit rate}, a larger area of the dark color part represents better block-sparse recovery performance.
While increasing the number of blocks, we find that block-sparse recovery methods (Block-ISTA and Ada-BlockLISTA) have better recovery performance than their non-block counterparts (ISTA and Ada-LISTA). 
The phenomenon is much more prominent, especially for the deep unfolding networks. 
The numerical results indicate that block-sparse recovery algorithms lead to better performance than the non-block counterparts when it comes to high noise power and a large number of blocks.
This is because our proposed network takes advantage of block structure and thus enjoys better recovery performances with a higher sparsity, which is also demonstrated in Theorem \ref{thm:convergence1}.
In our simulations, we validate the effectiveness of Ada-BlockLISTA which is not only superior to conventional ISTA and Block-ISTA in convergence speed but also exhibits stronger robustness to the block sparsity and measurement noise than \ac{adalista}.

\begin{figure} [tb]
	\centering
	\includegraphics[width=0.4\columnwidth]{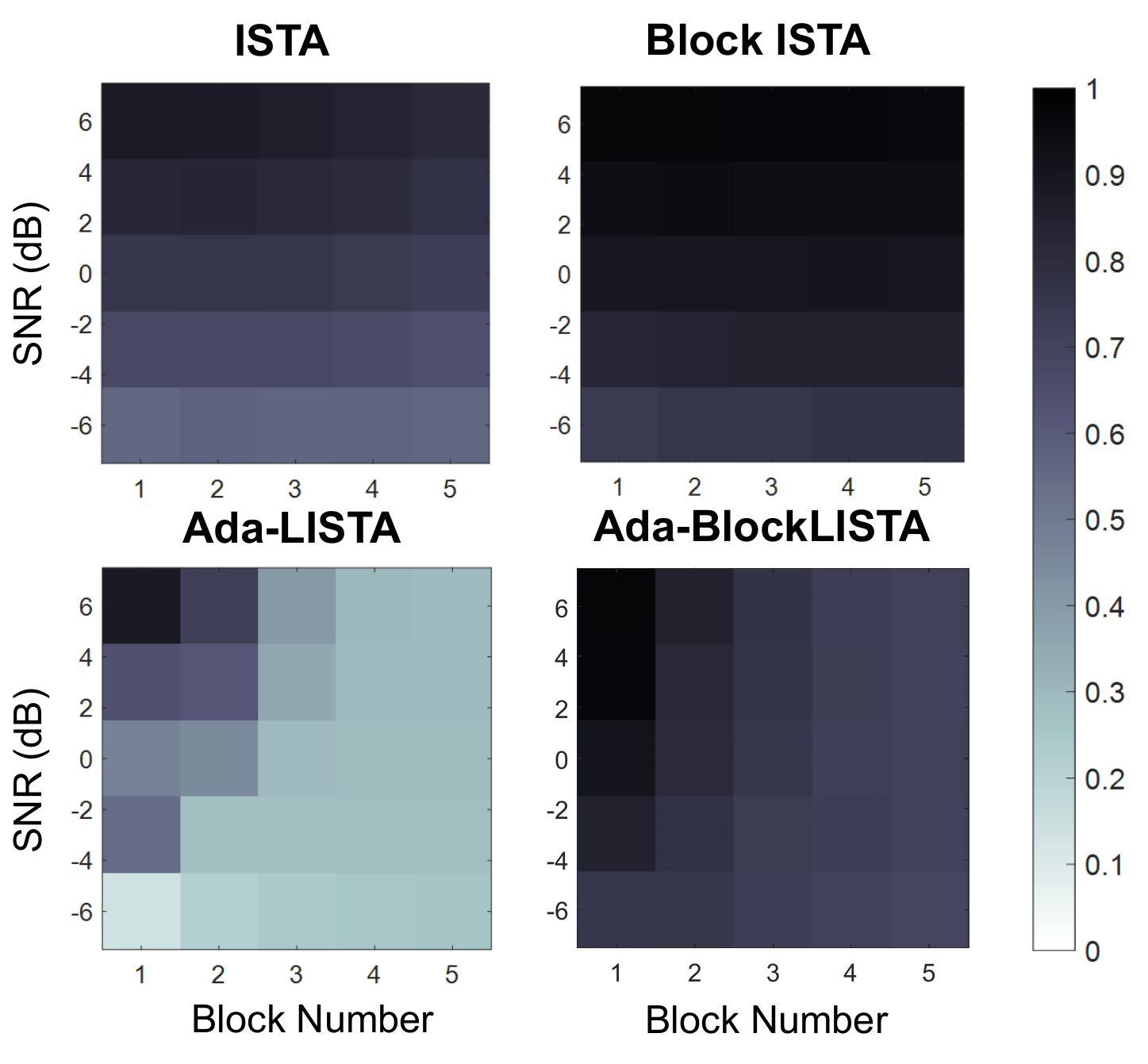}
	\caption{The hit rate of different algorithms in the noisy case.}
	\label{fig:hit rate}
\end{figure}

\section{Conclusion}\label{sec:conclusion}
In this work, we considered the block-sparse signal recovery problem, where nonzero entries of recovered signals occur in clusters.
We derived a block-sparse reconstruction network, named Ada-BlockLISTA, from \ac{adalista} by leveraging the particular block structure in the signal model. Furthermore, we prove that the unrolled block-sparse reconstruction network enjoys a linear convergence rate and provides a sufficient condition on block-sparsity to guarantee the successful recovery of block-sparse signals. 
Taking range-Doppler estimation of extended targets as an application, we analyzed its specific block structure and perform extensive simulations both in noiseless and noisy cases to verify the recovery performance of the proposed network.
The numerical results show that Ada-BlockLISTA yields better reconstruction properties compared to the original \ac{adalista} network, while increasing the number of blocks and noise power. 
By carrying out both theoretical and experimental analyses, we demonstrate that making explicit use of block-sparsity in unfolded deep learning networks enjoys improved block-sparse recovery performance. 
In the future, we will explore other additional structures in the signal model which helps build a structured unfolded network and find the optimal parameters.



\vspace{-0.4cm}
\bibliographystyle{elsarticle-num}
\bibliography{IEEEabrv}	

\begin{thebibliography}{10}
\expandafter\ifx\csname url\endcsname\relax
  \def\url#1{\texttt{#1}}\fi
\expandafter\ifx\csname urlprefix\endcsname\relax\def\urlprefix{URL }\fi
\expandafter\ifx\csname href\endcsname\relax
  \def\href#1#2{#2} \def\path#1{#1}\fi

\bibitem{2021Deep}
R.~Fu, V.~Monardo, T.~Huang, Y.~Liu, Deep unfolding network for block-sparse
  signal recovery, in: ICASSP 2021 - 2021 IEEE International Conference on
  Acoustics, Speech and Signal Processing (ICASSP), 2021, pp. 2880--2884.
\newblock \href {https://doi.org/10.1109/ICASSP39728.2021.9414163}
  {\path{doi:10.1109/ICASSP39728.2021.9414163}}.

\bibitem{4313110}
W.~{Xu}, B.~{Hassibi}, Efficient compressive sensing with deterministic
  guarantees using expander graphs, in: 2007 IEEE Information Theory Workshop,
  2007, pp. 414--419.
\newblock \href {https://doi.org/10.1109/ITW.2007.4313110}
  {\path{doi:10.1109/ITW.2007.4313110}}.

\bibitem{review}
P.~Indyk, Sparse recovery using sparse random matrices, in: A.~L{\'o}pez-Ortiz
  (Ed.), LATIN 2010: Theoretical Informatics, Springer Berlin Heidelberg,
  Berlin, Heidelberg, 2010, pp. 157--157.

\bibitem{Eldar2012}
Y.~C. Eldar, G.~Kutyniok, Compressed Sensing: Theory and Applications,
  Cambridge University Press, 2012.

\bibitem{5621984}
C.~R. {Berger}, Z.~{Wang}, J.~{Huang}, S.~{Zhou}, Application of compressive
  sensing to sparse channel estimation, IEEE Communications Magazine 48~(11)
  (2010) 164--174.

\bibitem{7952786}
J.~{Xiong}, W.~{Wang}, Sparse reconstruction-based beampattern synthesis for
  multi-carrier frequency diverse array antenna, in: 2017 IEEE International
  Conference on Acoustics, Speech and Signal Processing (ICASSP), 2017, pp.
  3395--3398.

\bibitem{Balakrishnan2004A}
R.~D. Balakrishnan, H.~M. Kwon, A new inverse problem based approach for
  azimuthal {DOA} estimation, in: Global Telecommunications Conference, 2003.
  GLOBECOM '03. IEEE, 2004, pp. 2187--2191 vol.4.

\bibitem{Compbeamf}
A.~Xenaki, P.~Gerstoft, K.~Mosegaard, Compressive beamforming, The Journal of
  the Acoustical Society of America 136 (2014) 260.

\bibitem{HuangCSFAR}
T.~{Huang}, Y.~{Liu}, Compressed sensing for a frequency agile radar with
  performance guarantees, in: 2015 IEEE China Summit and International
  Conference on Signal and Information Processing (ChinaSIP), 2015, pp.
  1057--1061.

\bibitem{Beck2009A}
A.~Beck, M.~Teboulle, A fast iterative shrinkage-thresholding algorithm for
  linear inverse problems, Siam J Imaging Sciences 2~(1) (2009) 183--202.

\bibitem{Combettes2006Signal}
P.~L. Combettes, V.~R. Wajs, Signal recovery by proximal forward-backward
  splitting, Multiscale Model Simul 4~(4) (2006) 1168--1200.

\bibitem{antonello2018proximal}
N.~Antonello, L.~Stella, P.~Patrinos, T.~van Waterschoot, Proximal gradient
  algorithms: Applications in signal processing, arXiv preprint
  arXiv:1803.01621.

\bibitem{ZhangG17b}
J.~Zhang, B.~Ghanem, {ISTA-Net}: Iterative shrinkage-thresholding algorithm
  inspired deep network for image compressive sensing, CoRR abs/1706.07929.
\newblock \href {http://arxiv.org/abs/1706.07929} {\path{arXiv:1706.07929}}.

\bibitem{Gregor2010Learning}
K.~Gregor, Y.~Lecun, Learning fast approximations of sparse coding, in:
  International Conference on International Conference on Machine Learning,
  2010, pp. 399--406.

\bibitem{AMP-Inspired}
M.~{Borgerding}, P.~{Schniter}, S.~{Rangan}, {AMP}-inspired deep networks for
  sparse linear inverse problems, IEEE Transactions on Signal Processing
  65~(16) (2017) 4293--4308.

\bibitem{OnsagerLAMP}
M.~Borgerding, P.~Schniter, Onsager-corrected deep learning for sparse linear
  inverse problems, in: 2016 IEEE Global Conference on Signal and Information
  Processing (GlobalSIP), 2016, pp. 227--231.

\bibitem{Fu2019}
R.~{Fu}, T.~{Huang}, Y.~{Liu}, Y.~C. {Eldar}, Compressed {LISTA} exploiting
  {Toeplitz} structure, in: 2019 IEEE Radar Conference (RadarConf), 2019, pp.
  1--6.

\bibitem{NEURIPS2018_cf8c9be2}
X.~Chen, J.~Liu, Z.~Wang, W.~Yin,
  \href{https://proceedings.neurips.cc/paper/2018/file/cf8c9be2a4508a24ae92c9d3d379131d-Paper.pdf}{Theoretical
  linear convergence of unfolded ista and its practical weights and
  thresholds}, in: S.~Bengio, H.~Wallach, H.~Larochelle, K.~Grauman,
  N.~Cesa-Bianchi, R.~Garnett (Eds.), Advances in Neural Information Processing
  Systems, Vol.~31, Curran Associates, Inc., 2018.
\newline\urlprefix\url{https://proceedings.neurips.cc/paper/2018/file/cf8c9be2a4508a24ae92c9d3d379131d-Paper.pdf}

\bibitem{liu2018alista}
J.~Liu, X.~Chen, Z.~Wang, W.~Yin,
  \href{https://openreview.net/forum?id=B1lnzn0ctQ}{{ALISTA}: Analytic weights
  are as good as learned weights in {LISTA}}, in: International Conference on
  Learning Representations, 2019.
\newline\urlprefix\url{https://openreview.net/forum?id=B1lnzn0ctQ}

\bibitem{2020Ada}
A.~Aberdam, A.~Golts, M.~Elad, Ada-lista: Learned solvers adaptive to varying
  models\href {http://arxiv.org/abs/2001.08456} {\path{arXiv:2001.08456}}.

\bibitem{ExtendedTarget}
L.~Wang, T.~Huang, Y.~Liu, Theoretical analysis for extended target recovery in
  randomized stepped frequency radars, arXiv preprint arXiv:1908.02929.

\bibitem{990897}
S.~F. {Cotter}, B.~D. {Rao}, Sparse channel estimation via matching pursuit
  with application to equalization, IEEE Transactions on Communications 50~(3)
  (2002) 374--377.

\bibitem{4550564}
F.~{Parvaresh}, H.~{Vikalo}, S.~{Misra}, B.~{Hassibi}, Recovering sparse
  signals using sparse measurement matrices in compressed dna microarrays, IEEE
  Journal of Selected Topics in Signal Processing 2~(3) (2008) 275--285.

\bibitem{Majumdar2010Compressed}
Majumdar, A., Ward, R., K., Compressed sensing of color images, Signal
  Processing 90 (2010) 3122--3127.
\newblock \href {https://doi.org/10.1016/j.sigpro.2010.05.016}
  {\path{doi:10.1016/j.sigpro.2010.05.016}}.

\bibitem{Block-Yonina}
Y.~C. Eldar, P.~Kuppinger, H.~Bolcskei, Block-sparse signals: Uncertainty
  relations and efficient recovery, IEEE Transactions on Signal Processing
  58~(6) (2010) 3042--3054.
\newblock \href {https://doi.org/10.1109/TSP.2010.2044837}
  {\path{doi:10.1109/TSP.2010.2044837}}.

\bibitem{Wu2020Sparse}
K.~Wu, Y.~Guo, Z.~Li, C.~Zhang, Sparse coding with gated learned {ISTA}, in:
  International Conference on Learning Representations, 2020.

\bibitem{bomp}
Y.~C. {Eldar}, P.~{Kuppinger}, H.~{Bolcskei}, Block-sparse signals: Uncertainty
  relations and efficient recovery, IEEE Transactions on Signal Processing
  58~(6) (2010) 3042--3054.

\bibitem{2010BCoSaMP}
R.~G. Baraniuk, V.~Cevher, M.~F. Duarte, C.~Hegde, Model-based compressive
  sensing, IEEE Transactions on Information Theory 56~(4) (2010) 1982--2001.

\bibitem{BlockBayesian}
Z.~{Zhang}, B.~D. {Rao}, Sparse signal recovery with temporally correlated
  source vectors using sparse bayesian learning, IEEE Journal of Selected
  Topics in Signal Processing 5~(5) (2011) 912--926.

\bibitem{Draganic2017On}
A.~Draganic, I.~Orovic, S.~Stankovic, On some common compressive sensing
  recovery algorithms and applications - review paper, Facta universitatis -
  series: Electronics and Energetics 30 (2017) 477--510.
\newblock \href {https://doi.org/10.2298/FUEE1704477D}
  {\path{doi:10.2298/FUEE1704477D}}.

\bibitem{2019Complex}
S.~Takabe, T.~Wadayama, Y.~C. Eldar, Complex field-trainable ista for linear
  and nonlinear inverse problems.

\bibitem{kingma2014adam}
D.~Kingma, J.~Ba, Adam: A method for stochastic optimization, Computer Science.

\bibitem{Zarka2020Deep}
J.~Zarka, L.~Thiry, T.~Angles, S.~Mallat,
  \href{https://openreview.net/forum?id=SJxWS64FwH}{Deep network classification
  by scattering and homotopy dictionary learning}, in: International Conference
  on Learning Representations, 2020.
\newline\urlprefix\url{https://openreview.net/forum?id=SJxWS64FwH}

\bibitem{4137843}
S.~R.~J. Axelsson, Analysis of random step frequency radar and comparison with
  experiments, IEEE Transactions on Geoscience and Remote Sensing 45~(4) (2007)
  890--904.
\newblock \href {https://doi.org/10.1109/TGRS.2006.888865}
  {\path{doi:10.1109/TGRS.2006.888865}}.

\bibitem{4338057}
S.~R. Axelsson, Analysis of ultra wide band noise radar with randomized stepped
  frequency, in: 2006 International Radar Symposium, 2006, pp. 1--4.
\newblock \href {https://doi.org/10.1109/IRS.2006.4338057}
  {\path{doi:10.1109/IRS.2006.4338057}}.

\bibitem{6212202}
T.~Huang, Y.~Liu, G.~Li, X.~Wang, Randomized stepped frequency isar imaging,
  in: 2012 IEEE Radar Conference, 2012, pp. 0553--0557.
\newblock \href {https://doi.org/10.1109/RADAR.2012.6212202}
  {\path{doi:10.1109/RADAR.2012.6212202}}.

\bibitem{Richardsbook}
M.~Richards, Fundamentals of Radar Signal Processing, 2005.

\bibitem{9205659}
T.~Huang, N.~Shlezinger, X.~Xu, D.~Ma, Y.~Liu, Y.~C. Eldar, Multi-carrier agile
  phased array radar, IEEE Transactions on Signal Processing 68 (2020)
  5706--5721.
\newblock \href {https://doi.org/10.1109/TSP.2020.3026186}
  {\path{doi:10.1109/TSP.2020.3026186}}.

\bibitem{2foldblockToep}
Y.~{Chi}, Y.~{Chen}, Compressive two-dimensional harmonic retrieval via atomic
  norm minimization, IEEE Transactions on Signal Processing 63~(4) (2015)
  1030--1042.

\end{thebibliography}

\newpage	
\appendix

\section{Convergence Proof} \label{sec:Proof}

In this appendix, we prove the linear convergence of Ada-BlockLISTA which has already been illustrated in Section  \ref{sec:Convergence}. 
In Appendix~\ref{ap:proof_lemma1}, we first prove a lemma that will be used in the proof of Theorem \ref{thm:convergence1}.
Then the detailed proof of Theorem \ref{thm:convergence1} is provided in Appendix \ref{ap:proof_Thm1}, consisting of two main steps.

\subsection{Proof of Lemma \ref{thm:lemma1}}
\label{ap:proof_lemma1}
We first prove the following lemma, which will be useful in bounding each Ada-BlockLISTA iteration:
\begin{lemma}\label{thm:lemma1}
	With  $\bm x,\bm x^*,\bm z \in {\mathbb{C}^M}$, we have
	\begin{equation}\label{eq:lemma}
	\norm{\bm x - \bm x^*}_2 \leq \theta + \norm{\bm z - \bm x^*}_2.
	\end{equation}
	where ${\bm x} = {\bm z}{\left( 1 - {\theta} / \left\| {\bm z} \right\|_2 \right)_{ + }}$.
\end{lemma}

\begin{proof}
	Consider the case where ${\theta} > \left\| {\bm z} \right\|_2$. 
	Then,
	$$\left( 1 - {\theta} / \left\| {\bm z} \right\|_2 \right) < 0, $$
	and the left side of the inequality is $\norm{ \bm x^*}_2$.
	Using the triangle inequality, one can obtain,
	\begin{align}
	\theta > \| \bm z \|_2 &= \| \bm{z} + \bm{x}^* - \bm{x}^*\|_2\\
	&\geq \| \bm{x}^* \|_2 - \| \bm{z} - \bm{x} \|_2.
	\end{align}
	Thus, $\norm{\bm x - \bm x^*}_2  = \left\| {\bm x^*} \right\|_2 \le {\theta} + \left\| {\bm z} - {\bm x^*} \right\|_2$.
	
	Next, consider the case where ${\theta} \le \left\| {\bm z} \right\|_2$. Then,
	$${\bm x} = {\bm z}{\left( 1 - {\theta} / \left\| {\bm z} \right\|_2 \right)}.$$ 
	Since $\| \bm z\| \geq \theta > 0,$ we can bound $\theta$ as,
	\begin{align*}
	\theta &= \frac{ \norm{\bm z \theta}_2} { \left\| {\bm z} \right\|_2 } \\
	&= \norm{ \bm x - \bm z }_2 \\
	&= \norm{ \left(\bm x - \bm x^*\right) - \left(\bm z - \bm x^*\right)}_2 \\
	& \geq \norm{\bm x - \bm x^*}_2 - \norm{\bm z - \bm x^*}_2 ,
	\end{align*}
	where the last inequality comes from the triangle inequality.
	Rearranging obtains the desired result.
\end{proof}

\subsection{Proof of Theorem \ref{thm:convergence1}}
\label{ap:proof_Thm1}

\begin{proof}
	
	Recall from the Theorem statement that, $\bm x^*\in \mathcal{X}(\zeta,s)$, and let $\bm \varepsilon \in \mathbb{C}^M$ be a complex standard random variable upper bounded by $\sigma$ with high probability.
	We define a subset $\mathcal{S}$ of cardinality $ \abs{\mathcal{S}} \le s$ containing all the indices of non-zero blocks in $\bm x^*$, i.e., if $q \in \mathcal{S}$, then $ \norm{\bm x^*_q}_2 > 0$. Otherwise $\bm x^*_q = \bm 0 $. 
	
	We separate the error bound in two parts:
	\begin{enumerate}
		\item 
		For zero blocks $\bm x^*_q,q \notin  \mathcal{S}$, we show that there are no false positive blocks in $ \bm x^{(t)}$ for all $t$, thus the recovery error of this part is $0$.
		\item
		For nonzero blocks $\bm x^*_q ,q \in  \mathcal{S}$, the norm of recovered error is bounded by mutual coherence.
	\end{enumerate}
	
	In the following proof, we use the notation $ \bm A_q^{(t)} = {\gamma}^{(t)}{\bm W}_q {\bf{\Phi }}_q$ for simplicity.
	
	\paragraph{Step 1} We prove the no-false-positive property by induction.
	The support hypothesis holds for $t = 0$, as we initial ${\bm x}^{(0)} $ as $\bm 0$. 
	From \eqref{eq:AdaBlockLISTA}, assuming $\bm x^{(t)}_q = \bm 0,q \notin \mathcal{S}$, we have
	\begin{align}
	{\bm z}_q^{(t+ 1)} 
	& = {\bm x}_q^{(t)} + \left({\bm A}_q^{(t)}\right)^H \left( \bm y - \sum\limits_{i = 1}^Q {{{\bm \Phi}_i}{\bm x}_i^{(t)}} \right) \nonumber \\
	& = \sum\limits_{i \in \mathcal{S}} \left({\bm A}_q^{(t)}\right)^H  {{\bm \Phi}_{i}} \left( {\bm x}_{i}^* - {\bm x}_{i}^{(t)} \right) + \left({\bm A}_q^{(t)}\right)^H  \bm \varepsilon.
	\end{align}
	
	Using the definition of block-coherence, the norm of this block $\left\| {\bm z}_q^{(t+ 1)} \right\|_2$ is bounded by $ P\tilde{\mu }_ B \sum\limits_{i \in \mathcal{S}}\left\| {\bm x}_{i}^* - {\bm x}_{i}^{(t)} \right\|_2 + C_{\bm W} \sigma$.
	
	As $ \theta ^{(t)} = P\tilde{\mu }_ B \sup\limits_{{\bm x}^*} \| {\bm x}^{(t)} - {\bm x}^* \|_{2,1}+ C_{\bm W} \sigma > \left\| {\bm z}_q^{(t+ 1)} \right\|_2$, we have
	\begin{eqnarray}\label{eq:nofalse}
	{\bm x}_q^{(t+ 1)} = {\bm z}_q^{(t+ 1)}{\left( 1 - {\theta ^{(t)}} / \left\| {\bm z}_q^{(t+ 1)} \right\|_2 \right)_{ + }} = \bm 0.
	\end{eqnarray}
	By induction, the no-false-positive property holds for all ${\bm x}_q^{(t)}, t \geq 0$.
	
	\paragraph{Step 2} We compute the error bound on the support set $\mathcal{S}$.
	
	For any nonzero block $\bm x^*_q ,q \in  \mathcal{S}$, the computation of ${\bm z}_q^{(t+ 1)}$ can be divided in three parts $(T1)$, $(T2)$ and $(T3)$, by separating index $q$ from $\mathcal{S}$
	\if false
	\begin{align}
	{\bm z}_q^{(t+ 1)} 
	& = \underbrace{ {\bm x}_q^{(t)} +   {\bf{\Phi }}_q^H \left({\bm W}_q^{(t)}\right)^H {{\bm \Phi}_q} \left( {\bm x}_q^* - {\bm x}_q^{(t)} \right) }_{(T1)} \nonumber \\
	&+ \underbrace{  \sum\limits_{q' = 1, q' \ne q}^Q {\bf{\Phi }}_q^H \left({\bm W}_q^{(t)}\right)^H  {{\bm \Phi}_{q'}} \left( {\bm x}_{q'}^* - {\bm x}_{q'}^{(t)} \right) }_{(T2)}  \notag \\
	& + \underbrace{  {\bf{\Phi }}_q^H \left({\bm W}_q^{(t)}\right)^H  \bm \varepsilon }_{(T3)},
	\end{align}
	\fi
	
	\begin{align}
	{\bm z}_q^{(t+ 1)} 
	& = \underbrace{ {\bm x}_q^{(t)} + \left({\bm A}_q^{(t)}\right)^H {{\bm \Phi}_q} \left( {\bm x}_q^* - {\bm x}_q^{(t)} \right) }_{(T1)} \nonumber \\
	&+ \underbrace{  \sum\limits_{q' = 1, q' \ne q}^Q \left({\bm A}_q^{(t)}\right)^H  {{\bm \Phi}_{q'}} \left( {\bm x}_{q'}^* - {\bm x}_{q'}^{(t)} \right) }_{(T2)}  \notag \\
	& + \underbrace{ \left({\bm A}_q^{(t)}\right)^H  \bm \varepsilon }_{(T3)},
	\end{align}
	We first consider part $(T1)$. 
	Since the diagonal elements of $\left({\bm A}_q^{(t)}\right)^H {{\bm \Phi}_q}$ are equal to $1$,
	we have
	\begin{eqnarray}
	(T1) = {\bm x}_q^* + \left(({\bm A}_q^{(t)})^H {{\bm \Phi}_q} - \bm I_P \right)( {\bm x}_q^* - {\bm x}_q^{(t)} ).
	\end{eqnarray}
	Then we move ${\bm x}_q^*$ to the left side and take the norm on both sides.
	
	According to Lemma \ref{thm:lemma1}, for all $q \in \mathcal{S}$, the norm of error can be computed as 
	\begin{align}
	\norm{ {\bm x}_q^{(t+ 1)} - \bm x_q^*}_2 
	& \leq \theta^{(t)} + \norm{\bm z_q^{(t+ 1)} - \bm x_q^*}_2 \nonumber \\
	& \leq \theta^{(t)} +  \left\| \left(({\bm A}_q^{(t)})^H {{\bm \Phi}_q} - \bm I_P \right)( {\bm x}_q^* - {\bm x}_q^{(t)} ) \right\|_2  \nonumber \\
	& + \sum\limits_{q' = 1, q' \ne q}^Q \left\| ({\bm A}_q^{(t)})^H  {{\bm \Phi}_{q'}} ( {\bm x}_{q'}^* - {\bm x}_{q'}^{(t)} ) \right\|_2  \notag \\
	& + C_{\bm W} \sigma.
	\end{align}
	
	Since $
	{\left\| ({\bm A}_q^{(t)})^H {\bm \Phi}_q - {{\bf{I}}_P} \right\|_s} \le \left( {P - 1} \right) \tilde{\nu }_ I
	$
	and
	$
	{\left\| ({\bm A}_q^{(t)})^H  {{\bm \Phi}_{q'}} \right\|_s} \le P\tilde{\mu }_ B
	$, we have
	\begin{align}\label{eq:norm_2}
	\norm{ {\bm x}_q^{(t+ 1)} - \bm x_q^*}_2 
	& \leq \theta^{(t)} +  \left( {P - 1} \right) \tilde{\nu }_ I \left\| {\bm x}_q^* - {\bm x}_q^{(t)} \right\|_2  \nonumber \\
	& + P\tilde{\mu }_ B \sum\limits_{q' = 1, q' \ne q}^Q \left\|  {\bm x}_{q'}^* - {\bm x}_{q'}^{(t)} \right\|_2  \notag \\
	& + C_{\bm W} \sigma.
	\end{align}
	
	As $\Supp({\bm x}^{(t)}) \subseteq \Supp (\bm x^*)$, we have
	\begin{align}\label{eq:norm_21}
	\norm{ {\bm x}^{(t+ 1)} - \bm x^*}_{2,1}
	& = \sum\limits_{q \in \mathcal{S}} \norm{ {\bm x}_q^{(t+ 1)} - \bm x_q^*}_{2} \nonumber \\	
	& \leq s\theta^{(t)} +  \left( {P - 1} \right) \tilde{\nu }_ I \sum\limits_{q \in \mathcal{S}} \left\| {\bm x}_q^* - {\bm x}_q^{(t)} \right\|_2  \nonumber \\
	& + P\tilde{\mu }_ B \sum\limits_{q \in \mathcal{S}} \sum\limits_{q' = 1, q' \ne q}^Q \left\|  {\bm x}_{q'}^* - {\bm x}_{q'}^{(t)} \right\|_2  \notag + sC_{\bm W} \sigma \nonumber \\
	& \leq s\theta^{(t)} +  \left( P - 1 \right) \tilde{\nu }_ I \left\| {\bm x}^* - {\bm x}^{(t)} \right\|_{2,1}  \nonumber \\
	& + P\tilde{\mu }_ B \left( s - 1 \right) \left\|  {\bm x}^* - {\bm x}^{(t)} \right\|_{2,1} + sC_{\bm W} \sigma.
	\end{align}
	
	Finally, by taking supremum over all $\bm x^* \in \mathcal{X}(\zeta,s)$ on both sides of \eqref{eq:norm_21}, we get
	
	\begin{align}\label{eq:norm_21_sup}
	\!\!\sup\limits_{{\bm x}^*} \!\! \norm{ {\bm x}^{(t+ 1)}\!\!  - \!\! \bm x^*}_{2,1}
	& \leq \left(  \left( P \!\! - \!\! 1 \right) \tilde{\nu }_ I \!\! +\!\!  P\tilde{\mu }_ B \left( s\!\!  -\!\!  1 \right)  \right) \!\! \sup\limits_{{\bm x}^*} \!\! \norm{ {\bm x}^* \!\! - \!\! {\bm x}^{(t)} }_{2,1} \nonumber \\
	& + s\theta^{(t)} + sC_{\bm W} \sigma.
	\end{align}
	
	As $ \theta ^{(t)} = P\tilde{\mu }_ B \sup\limits_{{\bm x}^*} \| {\bm x}^{(t)} - {\bm x}^* \|_{2,1}+ C_{\bm W} \sigma$, it is easy to find the recursive form for consecutive errors of $\sup\limits_{{\bm x}^*}  \norm{ {\bm x}^{(t+ 1)} - \bm x^*}_{2,1}$ and $\sup\limits_{{\bm x}^*}  \norm{ {\bm x}^{(t)} - \bm x^*}_{2,1}$ as below.
	\begin{align}\label{eq:recursive}
	& \sup\limits_{{\bm x}^*}  \norm{ {\bm x}^{(t+ 1)} - \bm x^*}_{2,1} \nonumber \\
	& \leq \left(  ( P - 1 ) \tilde{\nu }_ I + P\tilde{\mu }_ B ( 2s - 1 )  \right) \sup\limits_{{\bm x}^*} \left\| {\bm x}^{(t)} - \bm x^* \right\|_{2,1} + 2sC_{\bm W} \sigma 
	\nonumber \\
	& \leq \left(  ( P - 1 ) \tilde{\nu }_ I + P\tilde{\mu }_ B ( 2s - 1 )  \right)^{t+1} \sup\limits_{{\bm x}^*} \left\| {\bm x}^{(0)} - \bm x^* \right\|_{2,1} \nonumber \\	
	& + 2sC_{\bm W}\sigma \left(\sum\limits_{i = 0}^t 
	\left(  ( P - 1 ) \tilde{\nu }_ I + P\tilde{\mu }_ B ( 2s - 1 )  \right)^{i} 
	\right).
	\end{align}
	
	
	For all $\bm x^* \in \mathcal{X}(\zeta,s)$, $\left\| \bm x^* \right\|_{2,1} \le s\zeta $.
	Define two constants,
	\begin{align}
	c_1=-\log \left((  P \!\!-\!\! 1) \tilde{\nu }_ I \!\!+ \!\! P\tilde{\mu }_ B ( 2s \!\!-\!\! 1) \right) \\
	c_2=\frac{ 2s C_{\bm W}}{ 1 - ( P - 1 ) \tilde{\nu }_ I - P\tilde{\mu }_ B ( 2s - 1 )}.
	\end{align}
	
	If $s$ satisfies \eqref{eq:s}, then $( P - 1 ) \tilde{\nu }_ I + P\tilde{\mu }_ B ( 2s - 1 ) < 1$, 
	thus the network with parameters satisfying \eqref{eq:g_sub_coherence}, \eqref{eq:g_block_coherence} and \eqref{eq:C_W} can converge linearly and the recovered error bound for $k$-layer network holds 
	\begin{align}\label{eq:final}
	\sup\limits_{{\bm x}^*}  \norm{ {\bm x}^{(t+ 1)} - \bm x^*}_{2,1} \leq s \zeta \exp(-c_1(t+1)) +  c_2 \sigma .
	\end{align}
	which completes the proof of Theorem \ref{thm:convergence1}.
\end{proof}

\end{document}